\documentclass[11pt,a4paper]{article}
\usepackage[utf8]{inputenc} 
\usepackage[T1]{fontenc}
\usepackage{lmodern}

\usepackage[DIV=12]{typearea} 

\usepackage{microtype}  
\usepackage{array}
\usepackage{multirow}
\usepackage{amsmath} 
\usepackage{amssymb}
\usepackage{amsthm} 
\usepackage{thmtools}

\usepackage[procnumbered,ruled,vlined,linesnumbered]{algorithm2e}
 
\SetCommentSty{mycommfont} 

\usepackage{xcolor}  
\usepackage{xspace}
\usepackage{enumitem}    
\usepackage{tikz}   

\usepackage{comment} 

\usepackage{graphicx} 
\usepackage{caption} 
\usepackage{subcaption}  

\usepackage{authblk} 

\def\AdvCite{True} 

\ifdefined\AdvCite

\usepackage[
	style=alphabetic, 
	backref=true,
	doi=true,
	url=false,
	maxcitenames=3,
	mincitenames=3,
	maxbibnames=10,
	minbibnames=10,
	backend=biber,
	sortlocale=en_US
]{biblatex}

\usepackage[ocgcolorlinks]{hyperref} 
\usepackage[capitalise]{cleveref}

\DeclareCiteCommand{\citem}
{}
{\mkbibbrackets{\bibhyperref{\usebibmacro{postnote}}}}
{\multicitedelim}
{}

\DeclareCiteCommand{\citelink}
{}
{\bibhyperref{\usebibmacro{postnote}}}
{\multicitedelim}
{}

\renewcommand*{\multicitedelim}{\addcomma\space}

\AtEveryCitekey{\clearfield{note}}

\makeatletter
\AtBeginDocument{%
	\newlength{\temp@x}%
	\newlength{\temp@y}%
	\newlength{\temp@w}%
	\newlength{\temp@h}%
	\def\my@coords#1#2#3#4{%
		\setlength{\temp@x}{#1}%
		\setlength{\temp@y}{#2}%
		\setlength{\temp@w}{#3}%
		\setlength{\temp@h}{#4}%
		\adjustlengths{}%
		\my@pdfliteral{\strip@pt\temp@x\space\strip@pt\temp@y\space\strip@pt\temp@w\space\strip@pt\temp@h\space re}}%
	\ifpdf
	\typeout{In PDF mode}%
	\def\my@pdfliteral#1{\pdfliteral page{#1}}
	\def\adjustlengths{}%
	\fi
	\ifxetex
	\def\my@pdfliteral #1{}
	\def\adjustlengths{\setlength{\temp@h}{-\temp@h}\addtolength{\temp@y}{1in}\addtolength{\temp@x}{-1in}}%
	\fi%
	\def\Hy@colorlink#1{%
		\begingroup
		\ifHy@ocgcolorlinks
		\def\Hy@ocgcolor{#1}%
		\my@pdfliteral{q}%
		\my@pdfliteral{7 Tr}
		\else
		\HyColor@UseColor#1%
		\fi
	}%
	\def\Hy@endcolorlink{%
		\ifHy@ocgcolorlinks%
		\my@pdfliteral{/OC/OCPrint BDC}%
		\my@coords{0pt}{0pt}{\pdfpagewidth}{\pdfpageheight}%
		\my@pdfliteral{F}
		\my@pdfliteral{EMC/OC/OCView BDC}%
		\begingroup%
		\expandafter\HyColor@UseColor\Hy@ocgcolor%
		\my@coords{0pt}{0pt}{\pdfpagewidth}{\pdfpageheight}%
		\my@pdfliteral{F}
		\endgroup%
		\my@pdfliteral{EMC}%
		\my@pdfliteral{0 Tr}
		\my@pdfliteral{Q}%
		\fi
		\endgroup
	}%
} 
\makeatother

\else

\usepackage[ocgcolorlinks]{hyperref} 
\usepackage{cleveref}

\fi 

\allowdisplaybreaks[1]

\makeatletter
\g@addto@macro\bfseries{\boldmath}
\makeatother

\makeatletter
\g@addto@macro\mdseries{\unboldmath}
\g@addto@macro\normalfont{\unboldmath}
\g@addto@macro\rmfamily{\unboldmath}
\g@addto@macro\upshape{\unboldmath}
\makeatother

\colorlet{DarkRed}{red!50!black} 
\colorlet{DarkGreen}{green!50!black}
\colorlet{DarkBlue}{blue!50!black}

\hypersetup{
	linkcolor = DarkRed,
	citecolor = DarkGreen,
	urlcolor = DarkBlue,
	bookmarks = true,
	bookmarksnumbered = true,
	linktocpage = true 
}

\declaretheorem[numberwithin=section]{theorem}
\declaretheorem[numberlike=theorem]{lemma}
\declaretheorem[numberlike=theorem]{proposition}
\declaretheorem[numberlike=theorem]{corollary}
\declaretheorem[numberlike=theorem]{claim}

\crefname{algorithm}{Algorithm}{Algorithms}
\Crefname{algorithm}{Algorithm}{Algorithms}

\theoremstyle{definition}

\DontPrintSemicolon
\SetKw{KwAnd}{and}
\SetProcNameSty{textsc}
\SetFuncSty{textsc}

\newcommand{\ot}{\tilde{O}}

\newcommand{\ignore}[1]{}

\newcommand{\cC}{\mathcal{C}\xspace}

\newcommand{\eps}{\epsilon}

\newcommand{\poly}{\operatorname{poly}} 
\newcommand{\polylog}{\operatorname{polylog}}

\newcommand{\vol}{\operatorname{vol}}

\newcommand{\textlocal}{\operatorname{local}} 
\newcommand{\textpair}{\operatorname{pair}} 
\newcommand{\textout}{\operatorname{out}} 
\newcommand{\logk}{\lfloor \log_2k \rfloor}   
\newcommand{\logeta}{\lfloor \log_2\eta_i\rfloor}

\newcommand{\gap}{\gamma}%
 
\newcommand{\textbig}{\operatorname{big}} 
\newcommand{\textsmall}{\operatorname{small}} 
\newcommand{\out}{\operatorname{out}}  
\newcommand{\inn}{\operatorname{in}}  
\newcommand{\textin}{\operatorname{in}}

\newcommand{\Exp}{\operatorname{E}}  
\newcommand{\davg}{\overline{d}}          

\global\long\def\localEC{\text{LocalEC}}%
\global\long\def\DFS{\mathtt{DFS}}%

\ifdefined\ShowComment
\def\sebastian#1{\marginpar{$\leftarrow$\fbox{S}}\footnote{$\Rightarrow$~{\sf\textcolor{blue}{#1 --Sebastian}}}}
\def\thatchaphol#1{\marginpar{$\leftarrow$\fbox{T}}\footnote{$\Rightarrow$~{\sf\textcolor{purple}{#1 --Thatchaphol}}}}
\def\danupon#1{\marginpar{$\leftarrow$\fbox{D}}\footnote{$\Rightarrow$~{\sf\textcolor{orange}{#1 --Danupon}}}}
\def\sorrachai#1{\marginpar{$\leftarrow$\fbox{S}}\footnote{$\Rightarrow$~{\sf\textcolor{red}{#1 --Sorrachai}}}}
\def\note#1{#1}

\else
\def\sebastian#1{}
\def\thatchaphol#1{}
\def\danupon#1{}
\def\shen#1{}
\def\sorrachai#1{}
\def\note#1{} 

\fi

\title{Computing and Testing Small Connectivity in Near-Linear Time and Queries via Fast Local Cut Algorithms\thanks{This paper resulted
    from a merge of two papers submitted to arXiv \cite{ForsterY-arXiv19,NSY-arXiv19} and will be presented at the 31st Annual ACM-SIAM Symposium on Discrete Algorithms (SODA 2020).}}

\author[1]{Sebastian Forster\thanks{This author previously published under the name Sebastian Krinninger. Work partially done while at Max Planck Institute for Informatics and while at University of Vienna.}}
\author[2]{Danupon Nanongkai}
\author[3]{Thatchaphol Saranurak}
\author[4]{Liu Yang}%
\author[5]{Sorrachai Yingchareonthawornchai\thanks{Work partially
    done while at Michigan State University, USA.}}
\affil[1]{University of Salzburg, Austria}
\affil[2]{KTH Royal Institute of Technology, Sweden}
\affil[3]{Toyota Technological Institute at Chicago, USA}
\affil[4]{Independent Researcher}
\affil[5]{Aalto University, Finland}

\date{} 

\hypersetup{
	pdftitle = {Computing and Testing Small Connectivity in Near-Linear Time and Queries via Fast Local Cut Algorithms},
	pdfauthor = {Sebastian Forster, Danupon Nanongkai, Thatchaphol Saranurak, Liu Yang, Sorrachai Yingchareonthawornchai}
}                    

\begin{document}

		\maketitle
		
\begin{abstract}
	
Consider the following ``local'' cut-detection problem in a directed graph:
We are given a seed vertex $ x $ and need to remove at most $k$ edges so that at most $\nu$ edges can be reached from $x$ (a ``local'' cut)  or output $\bot$ to indicate that no such cut exists. If we are given query access to the input graph, then this problem can in principle be solved without reading the whole graph and with query complexity depending on $ k $ and $ \nu $.  In this paper we consider a slack variant of this problem where, when such a cut exists, we can output a cut with up to $O(k\nu)$ edges reachable from $x$.

We present a simple randomized algorithm spending $ O (k^2 \nu) $ time
and $ O (k \nu) $ queries for the above variant, improving in
particular a previous time bound of $ O (k^{O(k)} \nu) $ by
Chechik et al.~\citem[SODA '17]{ChechikHILP17}. We also extend our algorithm to handle an
approximate variant. 
We demonstrate that these local algorithms are versatile primitives for
designing substantially improved algorithms for classic graph
problems by providing the following three applications. (Throughout, $\tilde O(T)$ hides $\polylog(T)$.) 
\begin{enumerate} 
\item A randomized algorithm for the classic $k$-vertex connectivity
  problem that takes {\em near-linear} time when $k=O(\polylog(n))$,
  namely  $\tilde O(m+nk^3)$ time in undirected graphs. 
Prior to our work, the state of the art for this range of $k$ were
linear-time algorithms for $k \leq 3$ [\citelink[Tarjan FOCS '71]{Tarjan72}; \citelink[Hopcroft,
Tarjan SICOMP '73]{HopcroftT73}] and a recent algorithm with $\tilde
O(m+n^{4/3}k^{7/3})$ time \citem[Nanongkai~et~al.\ STOC '19]{NanongkaiSY19}. The story is
the same for directed graphs where our $\tilde O(mk^2)$-time
algorithm is near-linear when $k = O(\polylog(n))$. Our techniques
also yield an improved approximation scheme. 

\item Property testing algorithms for $k$-edge and -vertex
  connectivity with query complexities that are near-linear in $k$, \textit{exponentially} improving the state-of-the-art. 
This resolves two open problems, one by Goldreich and Ron~\citem[STOC '97]{GoldreichR02} and one by Orenstein and Ron~\citem[Theor.\ Comput.\ Sci.\ '11]{OrensteinR11}. 
\item A faster algorithm for computing the maximal $k$-edge connected
  subgraphs, improving prior work of Chechik et al.\ \citem[SODA '17]{ChechikHILP17}.
\end{enumerate}
\end{abstract}

		\tableofcontents    
		 
	\newpage 
	
	\section{Introduction}\label{sec:intro}

We study a natural ``local'' version of the problem of finding an edge cut of bounded size with applications to higher connectivity: given a seed vertex~$ x $, detect a ``small'' component containing~$ x $ with less than $ k $ outgoing edges (if it exists).
More precisely, given a vertex $x$ and two integers $\nu$ and $k$, our object of interest is a set $L\subseteq V$ such that 
\begin{align}\label{eq:intro:LocalEC}
x\in L,\ |E(L, V - L)|< k,\ \mbox{and}\ \vol^{\out}(L)\leq \nu\,,
\end{align} 
where $E(L, V - L)$ is the set of edges from vertices in $L$ to vertices outside of $L$, and $\vol^{\out}(L)$ is the outgoing volume of $ L $, the total number of outgoing edges for vertices in~$L$.
In particular, we will allow some slack in finding such a set.
For the problem we study, the
  input consists of  $x$, $k$, $\nu$, and a pointer to an adjacency-list representation of $G$, and an algorithm either
\begin{itemize}[noitemsep]
	\item returns a vertex set $S$ such that $|E(S, V - S)|< k$, or
	\item outputs that no $L\subseteq V$ satisfying \Cref{eq:intro:LocalEC} exists.
\end{itemize}

This problem, which we call LocalEC, has implicitly been studied before in the context of property testing~\cite{GoldreichR02,OrensteinR11} and for developing a centralized algorithm that computes the maximal $k$-edge connected subgraphs of a directed graph~\cite{ChechikHILP17}.
Recently, a variant for vertex cuts has been studied for obtaining faster vertex connectivity algorithms~\cite{NanongkaiSY19}.
A somewhat similar problem of locally detecting a small component with low-conductance has recently been studied extensively~\cite{KawarabayashiT19,HenzingerRW17,SaranurakW19}, in particular to obtain centralized algorithms for deterministically computing the edge connectivity of an undirected graph. 
In this paper, we present a simple randomized (Monte Carlo) algorithm for the problem above that takes $O(\nu k^2)$ time and makes $O(\nu k)$ edge queries. 
\begin{theorem}[Main Result]
	\label{thm:intro:LocalEC_approx}There is a randomized (Monte Carlo) algorithm
	that takes as input a vertex $x\in V$ of an $n$-vertex $m$-edge
	graph $G=(V,E)$ represented
	as adjacency lists, and integers $k$ and~$\nu$ such that $k <\nu<m/(130k)$, which accesses (i.e., makes queries for) $O(\nu k)$ edges and runs in $O(\nu k^2)$ time to output either (i) the symbol ``$\bot$'' to indicate that no such $L$ exists, or (ii) a set $S\subseteq V$ such that $|E(S, V - S)|< k$ in the following manner:
	\begin{itemize}[noitemsep,nolistsep]
		\item If $L\subseteq V$ satisfying \eqref{eq:intro:LocalEC} exists, the algorithm outputs a set $S$ as specified above with probability at least $3/4$. With probability at most $1/4$, it (incorrectly) outputs $\bot$.
		\item Otherwise, it outputs either a set $S$ as specified above or $\bot$. 
	\end{itemize}
\end{theorem}

Note that the error probability $1/4$ above can be made arbitrarily
small by repeating the algorithm. Our result in particular improves
the previous deterministic time bound of $ O (k^{O(k)} \nu) $ 
by Chechik et al.~\cite{ChechikHILP17}. 
It furthermore needs to be compared to a local variant~\cite{GoldreichR02} of Karger's randomized cut algorithm~\cite{Karger00} for \emph{undirected} graphs which finds a set $ L $ with $|E(S, V - S)|< k$ of size $ |L| \leq \sigma $ (where $\sigma$ is a given parameter) and volume at most $ \nu $ in time $ O(\sigma^{2 - 2/k} \nu) $ (if such a set exists).
Our algorithm is in fact very simple: it repeatedly finds a path starting at $x$ and ending at some random vertex. Our analysis is also very simple. 

Besides local computation algorithms being an object of study of
independent interest~\cite{RubinfeldTVX11,LeviM17}, the significance
of our contribution is illustrated  by
the fact that it almost readily implies faster algorithms for several problems in higher connectivity, as described below. 

\paragraph{(Global) Vertex Connectivity.} 
The main application of our result is efficient algorithms for the vertex connectivity problem. 
There has been a long line
of research on  this problem since at least five decades ago (e.g.~\cite{Kleitman1969methods,Podderyugin1973algorithm,EvenT75,Even75,Galil80,EsfahanianH84,Matula87,BeckerDDHKKMNRW82,LinialLW88,CheriyanT91,NagamochiI92,CheriyanR94,Henzinger97,HenzingerRG00,Gabow06,Censor-HillelGK14}).  (See Nanongkai~et~al.~\cite{NanongkaiSY19} for a more comprehensive literature survey.)
For the undirected case,  Aho,
Hopcroft and Ullman \cite[Problem 5.30]{AhoHU74} asked in their 1974 book for an
$O(m)$-time algorithm for
computing $\kappa_{G}$,  the vertex connectivity of graph $G$. Prior to our result, $O(m)$-time algorithms were known only when $\kappa_G \leq 3$, due to the classic results of Tarjan~\cite{Tarjan72} and Hopcroft and Tarjan~\cite{HopcroftT73}. In this paper, we present an algorithm that takes near-linear time whenever  $\kappa_G=O(\polylog(n))$. 

\begin{theorem}
	\label{thm:intro:VC_undir}There is a randomized (Monte Carlo) algorithm
	that takes as input an undirected graph and, with high probability, in time $\tilde{O}(m+n\kappa^{3})$
	outputs a vertex cut $S$ of size $\kappa$.\footnote{As usual, with high probability (w.h.p.) means with probability at least $1-1/n^c$ for an arbitrary constant $c\geq 1$.}
\end{theorem}

The above result is near-linear time whenever $\kappa=O(\polylog(n))$. 
By combining with previous results 
(e.g., \cite{HenzingerRG00,LinialLW88}),
the best running time for solving vertex connectivity in undirected graphs now is
$\tilde{O}(m+\min\{n\kappa^{3}, n^2\kappa, n^{\omega} + n\kappa^{\omega} \})$. Prior
to our work, the best running time for $\kappa>3$ was  $\tilde{O}(m+
\min\{n^{4/3}\kappa^{7/3}, n^2\kappa, n^{\omega} + n\kappa^{\omega} \} )$
\cite{NanongkaiSY19, HenzingerRG00, LinialLW88} where $\omega \leq
2.3728$ is the current matrix-multiplication exponent.  In particular, we have an improved running time when $\kappa\leq O(n^{0.457})$. 
Table~\ref{tab:results vertex connectivity} gives an overview of  all  these results.

\begin{table}[htbp]

	\centering

	\begin{tabular}{c c c c}
		\textbf{Graph class} & \textbf{Running time} & \textbf{Deterministic} & \textbf{Reference} \\
		directed & $ O ( \min \{ n^{3/4}, \kappa^{3/2} \} \cdot \kappa m + m n) $ & yes & \cite{Gabow06} \\
		directed & $ \tilde O (m n) $ & no & \cite{HenzingerRG00} \\
		directed & $ \tilde O (n^\omega + n \kappa ^\omega) $ & no & \cite{CheriyanR94} \\
		directed & $ \tilde O (\kappa \cdot \min \{ m^{2/3} n, m^{4/3} \}) $ & no &  \cite{NanongkaiSY19} \\
		directed & $ \tilde O  ( \kappa \cdot \min\{ n
                           \kappa^2 + m^{1/2}n \kappa^{1/2},m\kappa\}) $ & no & \textbf{here} \\
		undirected & $ O ( \min \{ n^{3/4}, \kappa^{3/2} \} \cdot \kappa^2 n + \kappa n^2 ) $ & yes & \cite{Gabow06} \\
		undirected & $ \tilde O (\kappa n^2) $ & no & \cite{HenzingerRG00} \\
		undirected & $ \tilde O (n^\omega + n \kappa ^\omega) $ & no & \cite{LinialLW88} \\
		undirected & $ \tilde O (m+ \kappa^{7/3} n^{4/3}) $ & no & \cite{NanongkaiSY19} \\
		undirected & $ \tilde O (m+\kappa^3 n ) $ & no &  \textbf{here}
	\end{tabular}
\caption{Comparison of algorithms for computing the vertex
  connectivity $ \kappa $.} 
	\label{tab:results vertex connectivity}
\end{table}

This result is obtained essentially by using our LocalEC algorithm to solve the {\em local vertex connectivity} problem (LocalVC) that was recently studied by Nanongkai, Saranurak, and Yingchareonthawornchai~\cite{NanongkaiSY19}, and then plugging this algorithm into the recent framework of~\cite{NanongkaiSY19}. The overall algorithm is fairly simple: Let $L$ be one side of the optimal vertex cut~$ M $, i.e., $ M $ consists of the neighbors of $ L $. We guess the values $\nu=\vol(L)$ and $k=\kappa$, and run our LocalVC algorithm with parameters $\nu$ and $k$ on $n/\nu$ randomly-selected seed vertices $x$. 

\paragraph{Approximation Algorithms and Directed Graphs.} The results in
\Cref{thm:intro:LocalEC_approx,thm:intro:VC_undir} can be generalized
to $(1+\epsilon)$-approximation algorithms and to algorithms on
directed graphs.  
The approximation guarantee means that the output vertex cut $S$ is of
size less than $\left\lfloor
  (1+\epsilon)k\right\rfloor$. 
In particular, the time complexity for locally computing a vertex cut (LocalVC) is $O(\nu k/\epsilon)$.\footnote{In \Cref{thm:localEC_approx_new}, we will state our result as an algorithm for the {\em gap} version of LocalVC, where the output $S$ in \Cref{thm:intro:LocalEC_approx} satisfies $|E(S, V - S)|<k+\gamma$ for some parameter $\gamma$. In this case, the time and query complexities decrease by a factor of $\gamma$.}
This improves the $\tilde O(\nu^{1.5}/(\sqrt{k}\epsilon^{1.5}))$-time algorithm of Nanongkai et~al.~\cite{NanongkaiSY19} when $ \nu \geq k^3 / \epsilon $. %
For computing $ \kappa $ exactly, our time complexity is $ \tilde O  ( \kappa \cdot \min\{ n \kappa^2 + m^{1/2}n \kappa^{1/2},m\kappa\}) $, and for approximating $ \kappa $, it is $\tilde{O}(\min\{m\kappa/\epsilon , n^{2+o(1)}\sqrt{\kappa}/\poly(\epsilon)\})$.

\paragraph{Testing Vertex- and Edge- Connectivity.} The study of testing graph properties, initiated by Goldreich et~al.~\cite{GoldreichGR98}, concerns the number of {\em queries} made to answer questions about graph properties.
In the {\em (unbounded-degree) incident-lists model} \cite{GoldreichR02,OrensteinR11}, it is assumed that there is a list $L_v$ of edges incident to each vertex $v$ (or lists of outgoing and incoming edges for directed graphs), and  an algorithm can make a query $q(v, i)$ for the $i^{th}$ edge in the list $L_v$ (if $i$ is bigger than the list size, the algorithm receives a special symbol in return). For any $\epsilon>0$, we say that an $m$-edge graph $G$ is {\em $\epsilon$-far} from having a property $P$ if the number of edge insertions and deletions to make $G$ satisfy $P$ is at least $\epsilon m$. {\em Testing $k$-vertex connectivity} is the problem of distinguishing between $G$ being $k$-vertex connected and $ G $ being $\epsilon$-far from having this property. Testing $k$-edge connectivity is defined analogously. It is assumed that the algorithm receives $n$, $\epsilon$, and $k$ in the beginning. We show the following. 
\begin{theorem}\label{thm:intro:vertex-testing}
In the unbounded-degree incident-list model, $k$-vertex and -edge
connectivity (where $k = \ot (\epsilon n)$) for directed graphs can be tested with $\ot(k/\epsilon^2)$ queries with probability at least $2/3$. %
 Further, $k$-edge connectivity for {\em simple} directed graphs can be tested with $\ot(\min\{k/\epsilon^2, 1/\epsilon^3 \}) $ queries. 
\end{theorem}
In particular, our  $\ot(k/\epsilon^2)$  bound is {\em linear} in $k$, and it can be independent of $k$ for testing $k$-edge connectivity on simple graphs. 
In the  {\em bounded-degree} incident-list model, the maximum degree $d$ is assumed to be given to the algorithm and a graph is said to be $\epsilon$-far from a property~$P$ if it needs at least $\epsilon nd$ edge modifications to have this property. We show the following. 
\begin{theorem} \label{thm:intro:vertex-testing-bounded}
In the bounded-degree incident-list model, $k$-vertex  and -edge
connectivity (where $k = \ot(\epsilon n)$) for directed graphs can be tested with  $\ot(k/\epsilon)$ queries with probability at least $2/3$. %
 Further, $k$-edge connectivity for {\em simple} directed graphs can be tested with $\ot(\min\{k/\epsilon,  1/\epsilon^2 \}) $ queries. 
\end{theorem}

It has been open for many years whether the bounds from
\cite{OrensteinR11,YoshidaI10,YoshidaI12} which are
exponential in $k$ can be made polynomial (this was asked in
e.g.~\cite{OrensteinR11}). We answer this question affirmatively with
bounds near-linear in $k$ using our local approximation algorithms. We
additionally solve the open problem of Goldreich and
Ron~\cite{GoldreichR02} of improving the local version of Karger's cut
algorithm for undirected graphs in the parameter range most relevant
to property testing.

\begin{table}[htbp]

	\centering

	\begin{tabular}{l c c}
		& \textbf{State of the art} & \textbf{Here} \\
		undirected $k$-edge connectivity & $ \tilde O \left( \frac{k^4}{(\epsilon \davg)^4} \right) $~\cite{ParnasR02} & $\ot\left(\frac{k^2}{\epsilon^2 \davg}\right) = \ot\left(\frac{k}{\epsilon^2}\right) $\\
		directed $k$-edge connectivity & $ \tilde O \left( \left( \frac{c k}{\epsilon \davg} \right)^{k+1} \right) $~\cite{YoshidaI10,OrensteinR11} & $\ot\left(\frac{k^2}{\epsilon^2 \davg}\right) = \ot\left(\frac{k}{\epsilon^2}\right) $\\
		directed $k$-edge connectivity & & $\ot\left(\min\{ \frac{k^2}{\epsilon^2 \davg}, \frac{k}{\epsilon^3 \davg} \}\right)$\\
		\hspace{1em} in simple graphs & & $= \ot\left(\min\{\frac{k}{\epsilon^2}, \frac{1}{\epsilon^3} \}\right)$\\
		(un)directed $k$-vertex connectivity & $ \tilde O \left( \left( \frac{c k}{\epsilon \davg} \right)^{k+1} \right) $~\cite{OrensteinR11} & $\ot\left(\frac{k^2}{\epsilon^2 \davg}\right) = \ot\left(\frac{k}{\epsilon^2}\right) $ 
	\end{tabular}
	\caption{Comparison of property testing algorithms for higher connectivity in unbounded-degree graphs.}
	\label{tab:results unbounded degree}
\end{table}

\begin{table}[htbp]

	\centering
	\begin{tabular}{l c c}
		& \textbf{State of the art} & \textbf{Here} \\
		undirected $3$-edge connectivity & $ O \left( \frac{\log (\frac{1}{\epsilon d})}{\epsilon^2 d} \right) $~\cite{GoldreichR02} & $ O \left( \frac{\log^2 (\frac{1}{\epsilon d})}{\epsilon} \right) $ \\
		undirected $k$-edge connectivity & $ \tilde O \left( \frac{k^3}{\epsilon^{3 - \frac{2}{k}} d^{2 - \frac{2}{k}}} \right) $~\cite{GoldreichR02} & $ \tilde O \left( \frac{k}{\epsilon} \right) $ \\
		directed $k$-edge connectivity & $ \tilde O \left( \left( \frac{c k}{\epsilon d} \right)^k d \right) $~\cite{YoshidaI10} & $ \tilde O \left( \frac{k}{\epsilon} \right) $ \\
		directed $k$-edge connectivity in simple graphs & & $\ot\left(\min\{\frac{k}{\epsilon},  \frac{1}{\epsilon^2} \}\right)$  \\
		undirected $k$-vertex connectivity & $ \tilde O \left( \left( \frac{c k}{\epsilon d} \right)^k d \right) $~\cite{YoshidaI12} & $ \tilde O \left( \frac{k}{\epsilon} \right) $ \\
		directed $k$-vertex connectivity & $ \tilde O \left( \left( \frac{c k}{\epsilon d} \right)^k d \right) $~\cite{OrensteinR11} & $ \tilde O \left( \frac{k}{\epsilon} \right) $
	\end{tabular}
	\caption{Comparison of property testing algorithms for higher connectivity in bounded-degree graphs.}
	\label{tab:results bounded degree}
\end{table}

\medskip \noindent {\em Detailed comparisons:} 
To precisely compare our bounds with the previous ones, note that there are two sub-models: (i) In the  {\em unbounded-degree} incident-list model, previous work assumes that $\davg=m/n$ is known to the algorithm in the beginning. (ii) in the  {\em bounded-degree} incident-list model, the maximum degree $d$ is assumed to be given to the algorithm and a graph is said to be $\epsilon$-far from a property $P$ if it needs at least $\epsilon nd$ edge modifications to have this property.
Our $\ot(k/\epsilon^2)$  bound can be generalized to a $\ot(k^2/(\epsilon^2 \davg))$  bound in the unbounded-degree model. Similarly, our  $\ot(\min\{k/\epsilon^2, 1/\epsilon^3 \})$ bound can be generalized to a $\ot(\min\{ k^2/(\davg\epsilon^2), k/(\davg\epsilon^3) \})$ bound in the unbounded-degree model.\footnote{\Cref{thm:intro:vertex-testing} can be obtained simply from the fact that $k\leq \davg $ and $ d \leq k/\epsilon$ can be assumed without loss of generality.}
The bounds that are exponential in $k$ by \cite{OrensteinR11,YoshidaI10,YoshidaI12} are $\tilde O((\frac{ck}{\epsilon \davg})^{k+1})$ and $\tilde O((\frac{ck}{\epsilon d})^{k}d)$ in the unbounded- and bounded-degree models, respectively, for testing both directed $k$-vertex and -edge connectivity. \Cref{tab:results unbounded degree,tab:results bounded degree} provide detailed comparisons between our and previous results.

\paragraph{Computing the Maximal $k$-Edge Connected Subgraphs.} For a set $ C \subseteq V $ of vertices, its induced subgraph~$ G [C] $ is a maximal $k$-edge connected subgraph of~$ G $ if $ G [C] $ is $k$-edge connected and no superset of $ C $ has this property.\footnote{This is related to the notion of edge strength by Bencz{\'{u}}r and Karger~\cite{BenczurK15} as follows: The strength of an edge $ e $ is defined as the maximum value of $ k $ such that some maximal vertex-induced $k$-edge connected subgraph contains both endpoints of $ e $.}
The problem of computing all maximal $k$-edge connected subgraphs of $ G $ is a natural generalization of computing the strongly connected components to higher edge connectivity.

For a long time, the state of the art for this problem was a running time of $ \ot (k m n) $ for $ k > 2 $ and $ O (m n) $ for $ k = 2 $ (implied by \cite{GabowT85} and~\cite{Gabow95}, respectively). %
The first improvement upon this was given by Henzinger, Krinninger, and Loitzenbauer with a running time of $ \ot (k^{O(1)} n^2) $ for $ k > 2 $ and $ O (n^2) $ for $ k = 2 $. %
The second improvement was given by Chechik et al.~\cite{ChechikHILP17} with a running time of $ \ot ((2k)^{k + 2} m^{3/2}) $ for $ k > 2 $ and $ O (m^{3/2}) $ for $ k = 2 $. %
In undirected graphs, a version of the algorithm by Chechik et al.\ runs in time $ \ot ((2k)^{k + 2} m \sqrt{n}) $ for $ k \geq 4 $ and in time $ O (m \sqrt{n}) $ for $ k \leq 3 $. %
In this paper, we improve upon this by designing an algorithm that has expected running time $ O (k^{3/2} m^{3/2}) $ in directed graphs, reducing the dependence on $ k $ from exponential to polynomial. %
We furthermore improve the running time for undirected graphs to $ O (k^3 n^{3/2} + k m) $, thus improving both the dependence on $ k $ \emph{and} on $ m $.
Table~\ref{tab:results connected subgraphs} compares our results to previous results. %

\begin{table}[htbp]

	\centering

	\begin{tabular}{c c c c c}
		\textbf{Parameter} & \textbf{Graph class} & \textbf{Running time} & \textbf{Deterministic} & \textbf{Reference} \\
		$ k = 2 $ & directed & $ O (m n) $ & yes & implied by~\cite{GabowT85} \\
		$ k \geq 3 $ & directed & $ \ot (k m n) $ & yes & implied by~\cite{Gabow95} \\ %
		$ k = 2 $ & directed & $ O (n^2) $ & yes & \cite{HenzingerKL15} \\
		$ k \geq 3 $ & directed & $ \ot (k^{O(1)} n^2) $ & yes & \cite{HenzingerKL15} \\ %
		$ k = 2 $ & directed & $ O (m^{3/2}) $ & yes & \cite{ChechikHILP17} \\ %
		$ k \geq 3 $ & directed & $ \ot ((2k)^{k + 2} m^{3/2}) $ & yes & \cite{ChechikHILP17} \\ %
		$ k \leq 3 $ & undirected & $ O (m \sqrt{n}) $ & yes & \cite{ChechikHILP17} \\ %
		$ k \geq 4 $ & undirected & $ \ot ((2k)^{k + 2} m \sqrt{n}) $ & yes & \cite{ChechikHILP17} \\ %
		$ k \geq 2 $ & directed & $ \ot (k^{3/2} m^{3/2}) $ & no & \textbf{here} \\ %
		$ k \geq 2 $ & undirected & $ \ot (k^3 n^{3/2} + k m) $ & no & \textbf{here} %
	\end{tabular}
	\caption{Comparison of algorithms for computing the maximal
          $k$-edge connected subgraphs.}
	\label{tab:results connected subgraphs}
\end{table}

Note that another natural way of generalizing the concept of strongly connected components to higher edge connectivity is the following:
A $k$-edge connected component is a maximal subset of vertices such that any pair of distinct vertices is $k$-edge connected in~$ G $.\footnote{Note that the $ k $ edge disjoint paths between a pair of vertices in a $k$-edge connected component might use edges that are not contained in the $k$-edge connected component. This is not allowed for maximal connected subgraphs.}
For a summary on the state of the art for computing the maximal $k$-edge connected subgraphs and components in both directed and undirected graphs, as well as the respective counterparts for vertex connectivity, we refer to~\cite{ChechikHILP17} and the references therein.

\paragraph{Organization.}  We provide basic graph notations in Section
\ref{sec:prelim}. We first prove the results for  local edge
connectivity in Section \ref{sec:localEC}. We then present three
applications. First, by developing an algorithm for local vertex connectivity in
Section \ref{sec:localVC}, we give near-linear time algorithms for small
vertex connectivity in Section \ref{sec:vertex_con_global}. Second, we
use approximate local vertex/edge connectivity to develop property
testing algorithms  for vertex/edge connectivity with near-linear in $k$ 
many edge queries in Section
\ref{sec:property-testing}. Finally, in Section \ref{sec:connected
  subgraphs}, we give improved algorithms for computing the
maximal $k$-edge connected subgraphs.

\section{Preliminaries}\label{sec:prelim}

Let $G=(V,E)$ be a \emph{directed} graph. For any $S,T\subseteq V$,
let $E(S,T)=\{(u,v)\mid u\in S,v\in T\}$. For each vertex $u$, we
let $\deg^{\out}(u) = | \{ v \in V \mid (u, v) \in E \} | $ denote the out-degree of $u$.
For a set $S\subseteq V$, the \emph{out-volume} of $S$ is $\vol^{\out}(S)=\sum_{u\in S}\deg^{\out} (u)$.
The set of \emph{out-neighbors} of~$S$ is $N^{\out}(S)=\{v\mid(u,v)\in E(S,V-S)\}$.
We define in-degree $\deg^{\inn}(u)$, in-volume $\vol^{\inn}(S)$,
and the set of in-neighbors $N^{\inn}(S)$ analogously. We add a subscript
$G$ to the notation when it is not clear which graph we are referring
to. 

A graph is \emph{strongly connected} if for every pair of vertices $ u $ and $ v $ there is a path from $ u $ to $ v $ and a path from $ v $ to $ u $. (The former implies the latter in undirected graphs and we usually just call the graph \emph{connected} in that case.)
A graph is \emph{$k$-edge connected} if it is strongly connected whenever fewer than $ k $ edges are removed.
A graph is \emph{$k$-vertex connected} if it has more than $ k $ vertices and is strongly connected whenever fewer than $ k $ vertices (and their incident edges) are removed.
The edge connectivity~$ \lambda_G $ of a graph $ G $ is the minimum value of $ k $ such that the graph is $k$-edge connected and the vertex connectivity~$ \kappa_G $ is the minimum value of $ k $ such that the graph is $k$-vertex connected.

We say that $(L,S,R)$ is a \emph{separation triple} of $G$ if $L,S,R$
partition $V$ such that $L,R\neq\emptyset$, and $E(L,R)=\emptyset$.
We also say that $S$ is a \emph{vertex cut} of $G$ of size $|S|$.
$S$ is an \emph{$st$-vertex cut} if $s\in L$ and $t\in R$. We
say that $s$ and $t$ are $k$-vertex connected
if there is no $st$-vertex cut of size less than $k$. Observe that $G$ is $k$-vertex connected
if and only if $s$ and $t$ are $k$-vertex connected for every pair $s,t\in V$.

	\section{Local Edge Connectivity}
\label{sec:localEC} 

In this section, we give a local algorithm for detecting an edge cut
of size $k$ and volume~$\nu$ containing some seed vertex in $O(\nu
k^{2})$ time. The algorithm accesses (i.e., make queries
for) $O(\nu k)$ edges (this bound is needed for our property
testing results). Roughly speaking, it outputs either a ``small'' cut or the
symbol ``$\bot$'' indicating that there is no small cut containing
$x$. The algorithm makes one-sided errors -- it might be wrong when it
outputs $\bot$ -- %
with probability at most $1/4$. By standard arguments, we can make
the error probability arbitrarily small by executing the algorithm repeatedly.
Both the algorithm description and the analysis are very simple.

\begin{theorem}
\label{thm:localEC_approx_new}
There exists the following randomized algorithm. It takes as inputs, 
\begin{itemize}[noitemsep,nolistsep]
\item a pointer to a seed vertex $x \in V$ in an adjacency list
  representing  an $n$-vertex $m$-edge  directed graph $G=(V,E)$, 
\item a volume parameter (positive integer) $\nu$, 
\item a cut-size parameter (positive integer) $k$, and 
\item a slack parameter (non-negative integer) $\gap$, where
\end{itemize}
\begin{align} \label{eq:precon_ec}
 \gamma \leq k < \nu < m(\gap+1)/(130k).
\end{align}
It accesses (i.e., makes queries for) $O(\nu k /(\gap+1))$ edges and runs in $O(\nu k^2/ (\gap+1))$ time.
in the following manner:
\begin{itemize}[noitemsep,nolistsep]
\item If there exists a vertex-set $S'$ such that $S' \ni x, \vol^{\out}(S')
  \leq \nu, $ and $ |E(S',V-S')| < k$, then with probability at least $3/4$, the algorithm
  outputs a non-empty vertex-set $S\subsetneq V$
  such that $|E(S,V-S)| < k+\gap$ and $\vol^{\out}(S) \leq 130\nu k/(\gap+1)$  (otherwise it outputs $\bot$).  
\item Otherwise (i.e., no such $S'$ exists), the algorithm outputs either
  a set $S$ as specified above or $\bot$. 
\end{itemize}
\end{theorem}
\paragraph{Remark.} If we are not concerned about the number of edge
queries, then there is an another simple algorithm with the same
running time in Appendix \ref{sec:localEC_alternate}. 

In particular, we obtain exact and $(1+\epsilon)$-approximate local
algorithms for \Cref{thm:localEC_approx_new} as follows. We set $\gap
= 0$ in \Cref{thm:localEC_approx_new} for the exact local algorithm.
For $(1+\epsilon)$-approximation, we set $\gap = \lfloor \epsilon k
\rfloor$.

\begin{corollary}
\label{cor:localEC_exact_new}
There exists the following randomized algorithm. It takes as the same inputs as
in \Cref{thm:localEC_approx_new} but with $\gap = 0$.  
It accesses (i.e., makes queries for) $O(\nu k)$ edges and
runs in $O(\nu k^2)$ time.  It then outputs 
in the following manner.  
\begin{itemize}[noitemsep,nolistsep]
\item If there exists a vertex-set $S'$ such that $S' \ni x, \vol^{\out}(S')
  \leq \nu, $ and $ |E(S',V-S')| < k$, then with probability at least $3/4$, the algorithm
  outputs a non-empty vertex-set $S\subsetneq V$
  such that $|E(S,V-S)| < k $  (otherwise it outputs $\bot$).  
\item Otherwise (i.e., no such $S'$ exists), the algorithm outputs either
  a set $S$ as specified above or $\bot$. 
\end{itemize}
\end{corollary}
 
\begin{corollary}
\label{cor:localEC_approx_new}
There exists the following randomized algorithm. It takes as the same inputs as
in \Cref{thm:localEC_approx_new} but with an additional parameter $\epsilon
\in (0,1]$ and $\gap = \lfloor \epsilon k \rfloor$.  
It accesses (i.e., makes queries for) $O(\nu /\epsilon)$ edges and
runs in $O(\nu k/ \epsilon)$ time.  It then outputs 
in the following manner.  
\begin{itemize}[noitemsep,nolistsep]
\item If there exists a vertex-set $S'$ such that $S' \ni x, \vol^{\out}(S')
  \leq \nu, $ and $ |E(S',V-S')| < k$, then with probability at least $3/4$, the algorithm
  outputs a non-empty vertex-set $S\subsetneq V$
  such that $|E(S,V-S)| < \lfloor (1+\epsilon) k \rfloor $  (otherwise it outputs $\bot$).  
\item Otherwise (i.e., no such $S'$ exists), the algorithm outputs either
  a set $S$ as specified above or $\bot$. 
\end{itemize}
\end{corollary}
\begin{proof}
The results follow from \Cref{thm:localEC_approx_new} where we set
$\gap = \lfloor \epsilon k \rfloor$, and the following fact: 
\begin{align} \label{eq:gap-to-approx}
k+\gap = k + \lfloor \epsilon k \rfloor  =\lfloor k \rfloor + \lfloor
  \epsilon k \rfloor  \leq \lfloor k+\epsilon k\rfloor = \lfloor (1+\epsilon) k\rfloor.
\end{align}\qedhere
\end{proof}

\begin{algorithm}
\SetKwFor{RepTimes}{repeat}{times}{end}
\KwIn{$\text{Seed vertex } x \in V, \text{target volume } \nu > k,
  \text{target cut-size } k \geq 1,  \text{slack }\gap \leq k$.} 
\KwOut{A vertex-set $S$ or the symbol $\perp$ as specified in \Cref{thm:localEC_approx_new}.}     
\BlankLine
\RepTimes{$k + \gap$}
{
$y \gets \mathtt{NIL}$. \;
Grow a DFS tree $T$ of vertices reachable from the seed vertex $x$ as follows: \;
\While{\normalfont{the DFS algorithm still has an edge $e = (a,b)$ to explore}} 
{
  \If{$e$ \normalfont{is} not marked} {
       Mark $e$.\;  
       \lIf{\normalfont the algorithm has marked $\geq 128\nu k/(\gap+1)$ edges}{
       	\Return{$\bot$.} \label{line:earlybot}
      \; With probability $(\gap+1)/(8\nu)$, set $y \gets a$ and \textbf{break the while-loop}. \label{line:break decision}
    }
  }
 }
\lIf{$y = \mathtt{NIL}$} { 
  \Return $S = V(T)$. \label{line:dfsstuck}
} \lElse  {Reverse the direction of edges on the path $P_{xy}$ in $T$ from $x$ to $y$.}
}

\Return{$\bot$.} \label{line:latebot} 

\caption{$\localEC(x \in V,\nu,k,\gap)$\label{alg:local_ec_new}}
\end{algorithm}

The algorithm for \Cref{thm:localEC_approx_new} is described in
\Cref{alg:local_ec_new}.  
Roughly speaking, in each iteration of the repeat-loop the algorithm runs a standard depth-first search (DFS)\footnote{The choice of depth-first search is arbitrary. Breadth-first search would also work.}
algorithm, but randomly stops before finishing the DFS; this stopping corresponds to breaking the while-loop in \Cref{alg:local_ec_new}. (If the DFS
finishes before being stopped, our algorithm outputs all vertices found by the DFS.) Let $(a,b)$ denote the last edge explored by the
DFS before the random stop happens. We reverse the direction of all
edges on the unique path from $x$ to $a$ in the DFS tree. We repeat
this whole process for $k+\gap$ iterations.
Additionally, we keep track of the number of edges visited in total over all iterations, which affects the query complexity.
If this number exceeds a certain bound, we stop the algorithm.

It is easy to see to bound the running time and the query complexity of \Cref{alg:local_ec_new}.

\begin{lemma} \label{lem:running-time-main}
\Cref{alg:local_ec_new} runs in time $O(\nu k^2/(\gap+1))$,
and accesses $O(\nu k/(\gap+1))$ edges. 
\end{lemma}
\begin{proof}
By line~\ref{line:earlybot}, the algorithm accesses at most $\lceil 128
\nu k /(\gap +1) \rceil = O(\nu k/(\gap+1)) $ edges in total.
For every vertex $ v $, we store all outgoing edges accessed so far in a doubly linked list.
As long as not all outgoing edges of $ v $ have yet been accessed, we additionally store the index of the next edge to query from the adjacency list of $ v $.
In addition, we store a pointer to the position in the doubly linked list of the last outgoing vertex of $ v $ visited in the current DFS.
Whenever we reverse an edge $ (v, w) $, we remove $ (w, v) $ from the list of $ v $ and append $ (w, v) $ to the list of $ w $.
Note that we only reverse edges $ (v, w) $ for which $ w $ was the last outgoing vertex of $ v $ visited in the current DFS.
from each vertex $ v $ to the last outgoing vertex of $ v $ visited in
the current DFS.
Therefore the position of each edge to be reversed in its doubly linked list can be found in constant time by using the corresponding pointers.
Thus, both the DFS and edge-reversal step take time at most linear in the total number of
edges accessed so far.
As we repeat this for at most $k+\gap \leq 2k$ iterations, the running time follows.
\end{proof}

 We start the correctness proof with the following important observation from~\cite{ChechikHILP17}. 

\begin{lemma}
\label{lem:reverse_cutsize_new} Let $S\subsetneq V$ be any set with $x\in S$.
Let $P_{xy}$ be a path from $x$ to some vertex $y$. Suppose we reverse the
direction of all edges on $P_{xy}$. Then, we have that $|E(S,V-S)|$ and $\vol^{\out}(S)$
are both decreased exactly by one if $y\notin S$. Otherwise, $|E(S,V-S)|$
and $\vol^{\out}(S)$ stay the same. 
\end{lemma}
\begin{proof}
We fix the set $S$ and the path $P_{xy}$ where $x \in S$. If $y \notin
S$, then $P_{xy}$ crosses the edges between $S$ and $V-S$ back and forth
so that the number of crossing outgoing
edges (i.e., from $S$ to $V-S$) is one plus the number of crossing incoming
edges (i.e., from $V-S$ to~$S$). Therefore, reversing the direction of the
edges in $P_{xy}$ decreases both $|E(S,V-S)|$ and $\vol^{\out}(S)$
exactly by one.  If $y \in S$, then the number crossing outgoing edges is
equal to the number of crossing incoming edges. Thus, reversing the
direction of edges in $P_{xy}$ does not change $|E(S,V-S)|$ and
$\vol^{\out}(S)$. 
\end{proof}

The following two lemmas prove the correctness of \Cref{alg:local_ec_new}.
\begin{lemma}
If a vertex set $S$ is returned, then  $\emptyset \neq S \subsetneq V,
|E(S,V-S)|< k + \gap$ and $\vol^{\out}(S) \leq 130\nu k/(\gap+1)$. 
\end{lemma} 
\begin{proof}
If a vertex set $S$ is returned, then the algorithm terminates in
line~\ref{line:dfsstuck} and $ S = V(T) $ (for the DFS tree~$ T $) is the set of vertices reachable from $ x $.
Thus, we have  $|E(S,V-S)|=0$ and $
\vol^{\out}(S) \leq 128\nu k/(\gap+1) \stackrel{(\ref{eq:precon_ec})}  < m $.
Note that by design $ x \in S $ (and thus $ S \neq \emptyset $), and since $ \vol^{\out}(S) < m $, also $ S \subsetneq V $.
 The algorithm has reversed strictly less than $k+\gap $ many paths $P_{xy}$ because
the algorithm did not reverse a path in the final iteration that returns $S$. Therefore,
\Cref{lem:reverse_cutsize_new} implies that, initially, $|E(S,V-S)|< k
+ \gap,$ and $\vol^{\out}(S)  < 128\nu k/(\gap+1)  + k+\gap \leq 130\nu k/(\gap+1) $.
\end{proof} %

\begin{lemma}
If there is a set $S'$ such that $x \in S', |E(S',V-S')|<k$ and
$\vol^{\out}(S')\le\nu$, then $\bot$ is returned with probability at
most $1/4$. 
\end{lemma}
\begin{proof}

Let $\tau = 128\nu k/(\gap+1)$, let $R_{\bot}^{(\ref{line:earlybot})}$ be the event that $\bot$ is returned in line~\ref{line:earlybot} and let $R_{\bot}^{(\ref{line:latebot})}$ be the event that $\bot$ is returned in line~\ref{line:latebot}.
We will show that each of these events will occur with probability at most $ 1 / 8 $.
As there are no other ways of returning $ \bot $ in the algorithm and the two events are mutually exclusive, it will follow that the probability of returning $ \bot $ is at most $ 1/8 + 1/8 = 1/4 $.

Note that event $R_{\bot}^{(\ref{line:earlybot})}$ only occurs if the decision to break the while loop in line~\ref{line:break decision} has been made less than $ k + \gap $ times before the limit of $ \tau $ marked edges is reached. 
We view the probabilistic decision to break the loop in line~\ref{line:break decision} as a Bernoulli trial with success probability $ p = (\gap+1)/(8\nu) $.
Let $ X $ be the random variable that denotes the number of successes in a sequence of $ \tau - 1 $ such trials with success probability $ p $.
Clearly, $ \Pr [R_{\bot}^{(\ref{line:earlybot})}] = \Pr [X < k + \gap] $.
Now consider an \emph{infinite} sequence of Bernoulli trials with success probability $ p $ and let $ Y $ be the random variable that denotes the number of trials needed to obtain $ k + \gap $ successes in this sequence.
Observe that $ \Pr [X < k + \gap] = \Pr [Y > \tau - 1] = \Pr [Y \geq \tau] $.
The distribution of $ Y $ is exactly the negative binomial distribution with parameters $ k + \gap $ and $ p $ (see, e.g.,~\cite{Feller68}) and thus $ \Exp [Y] = (k + \gap) / p $.
By Markov's inequality, $ \Pr [Y \geq \tau] = \Pr [Y \geq 8 \cdot (k + \gap) / p] = \Pr [Y \geq 8 \cdot \Exp [Y]] \leq 1 / 8 $.
Thus, $ \Pr [R_{\bot}^{(\ref{line:earlybot})}] = \Pr [Y \geq \tau] \leq 1/8 $.

To show that $P [R_{\bot}^{(\ref{line:latebot})}] \leq 1/8 $, recall first that the events $ R_{\bot}^{(\ref{line:latebot})} $ and $ R_{\bot}^{(\ref{line:earlybot})} $ are mutually exclusive and thus $P [R_{\bot}^{(\ref{line:latebot})}] = P [R_{\bot}^{(\ref{line:latebot})} \mid \bar{R}_{\bot}^{(\ref{line:earlybot})}] $, where $ \bar{R}_{\bot}^{(\ref{line:earlybot})} $ is the counter-event of $ R_{\bot}^{(\ref{line:earlybot})} $.
To prove our claim, we will show that $ P
[\bar{R}_{\bot}^{(\ref{line:latebot})} \mid
\bar{R}_{\bot}^{(\ref{line:earlybot})}] \geq 1 - 1/8 $, where $ \bar{R}_{\bot}^{(\ref{line:latebot})} $ is the counter-event of $ R_{\bot}^{(\ref{line:latebot})} $.
We first argue that, conditioned on $
\bar{R}_{\bot}^{(\ref{line:earlybot})} $, $
\bar{R}_{\bot}^{(\ref{line:latebot})} $ is implied by the following
event $ Q $: the algorithm either performs at most $ k + \gap - 1 $
iterations or the algorithm performs $ k + \gap $ iterations and for
at most $ \gap $ of these iterations $ y $ is set to a vertex in~$ S' $
in line~\ref{line:break decision}. If this is true, then we have $ \Pr
[\bar{R}_{\bot}^{(\ref{line:latebot})} \mid
\bar{R}_{\bot}^{(\ref{line:earlybot})}] \geq \Pr [Q \mid
\bar{R}_{\bot}^{(\ref{line:earlybot})}] $.

If the algorithm performs only $ k + \gap - 1 $ iterations, without the
occurrence of event $R_{\bot}^{(\ref{line:earlybot})}$, it must return
some set $ S = V(T) $ in line~\ref{line:dfsstuck}, which implies event
$ \bar{R}_{\bot}^{(\ref{line:latebot})} $ as desired.
Now assume that the algorithm performs $ k + \gap $ iterations and for
at most $ \gap $ of these iterations $ y $ is set to a vertex in~$ S'
$. We show that this situation implies event $ \bar{R}_{\bot}^{(\ref{line:latebot})} $.
The assumption implies that during the first $ k + \gap - 1 $ iterations, the algorithm must have set $ y $ to a vertex in $ V - S' $ in at least $ k - 1 $ iterations.
In each of these $ k - 1 $ iterations, $ | E (S', V - S') | $ is reduced by exactly one by \Cref{lem:reverse_cutsize_new} and initially $ | E (S', V - S') | \leq k - 1 $.
Therefore the algorithm has set $ y $ to a vertex in $ V - S' $ in \emph{exactly} $ k - 1 $ iterations.
Now at the beginning of iteration $ k + \gap $ (the final iteration) we have $ | E (S', V - S') | = 0 $.
It is therefore not possible that $ y $ is set to a vertex in $ V - S' $ in that iteration as these vertices are not reachable anymore.
It is furthermore not possible that $ y $ is set to a vertex in~$ S' $ as the number of times this already happened in previous iterations must be $ \gamma $ and for event $ Q $ to occur cannot increase to $ \gamma + 1 $.
Therefore, the only option left is that $ y $ remains $ \mathtt{NIL} $ in that final iteration and the algorithm returns $ S = V(T) $ in line~\ref{line:dfsstuck}, which implies event $ \bar{R}_{\bot}^{(\ref{line:latebot})} $.

Since $ \Pr [\bar{R}_{\bot}^{(\ref{line:latebot})} \mid
\bar{R}_{\bot}^{(\ref{line:earlybot})}] \geq \Pr [Q \mid
\bar{R}_{\bot}^{(\ref{line:earlybot})}] $, it remains to show that $ \Pr [Q \mid \bar{R}_{\bot}^{(\ref{line:earlybot})}] \geq 1 - 1/8 $.
Let $ Z $ be the random variable denoting the number of iterations for which $ y $ is set to a vertex in~$ S' $ in Line~\ref{line:break decision}.
Now, conditioned on $ \bar{R}_{\bot}^{(\ref{line:earlybot})} $, $ Q $ simply is the event that $ Z \leq \gap $.
Note that for each newly marked edge, $ y $ is set to the tail of this edge with probability $ p = (\gap+1)/(8\nu)$.
Furthermore, we only mark edges that have not been reversed yet.
Thus, $ y $ can be set to a vertex in $ S' $ only if the newly marked edge is from $ E (S', V) $ (where we refer to the state of this set at the beginning of the algorithm).
As we never unmark edges, the probability for each edge from $ E (S', V) $ to have $ y $ set to its tail is at most $ p $.
Therefore, $ \Exp [Z \mid \bar{R}_{\bot}^{(\ref{line:earlybot})}] \leq | E (S', V) | \cdot p = \nu p = (\gap+1)/8 $.
By Markov's inequality, we have $\Pr [Z < 8 \cdot \Exp [Z \mid \bar{R}_{\bot}^{(\ref{line:earlybot})}] \mid \bar{R}_{\bot}^{(\ref{line:earlybot})}] \geq 1-1/8$, and since $ \gap + 1 \geq 8 \cdot \Exp [Z \mid \bar{R}_{\bot}^{(\ref{line:earlybot})}] $ we get $ \Pr [Z < \gap + 1 \mid \bar{R}_{\bot}^{(\ref{line:earlybot})}] \geq \Pr [Z < 8 \cdot \Exp [Z \mid \bar{R}_{\bot}^{(\ref{line:earlybot})}] \mid \bar{R}_{\bot}^{(\ref{line:earlybot})}] $.
Finally, note that $ Z $ and $ \gap $ are integer and therefore $ \Pr [Z \leq \gap \mid \bar{R}_{\bot}^{(\ref{line:earlybot})}] = \Pr [Z < \gap + 1 \mid \bar{R}_{\bot}^{(\ref{line:earlybot})}] \geq \Pr [Z < 8 \cdot \Exp [Z \mid \bar{R}_{\bot}^{(\ref{line:earlybot})}] \mid \bar{R}_{\bot}^{(\ref{line:earlybot})}] \geq 1 - 1/8 $.
%
%
%
\end{proof}

\section{Local Vertex Connectivity}
\label{sec:localVC}

In this section, we give the vertex cut variant of the local algorithms
from \Cref{sec:localEC}.

\begin{theorem}
\label{thm:localVC_gap}
There exists the following randomized algorithm. It takes as inputs, 
\begin{itemize}[noitemsep,nolistsep]
\item a pointer to a seed vertex $x \in V$ in an adjacency list
  representing  an $n$-vertex $m$-edge directed graph $G=(V,E)$
\item a volume parameter (positive integer) $\nu$, 
\item a cut-size parameter (positive integer) $k$, and 
\item a slack parameter (non-negative integer) $\gap$, where
\end{itemize}
\begin{equation} \label{eq:precon_vc}
1 \leq k < n/4, \quad \mbox{ and } \quad \gamma \leq k < \nu < m(\gap+1)/(12480k).
\end{equation}
It accesses (i.e., makes queries for) $O(\nu k /(\gap+1))$ edges and runs in $O(\nu k^2/ (\gap+1))$ time. It then outputs 
in the following manner. 
\begin{itemize}[noitemsep,nolistsep]
\item If there exists a separation triple $(L',S',R')$ such that $L' \ni x, \vol^{\out}(L')
  \leq \nu, $ and $|S'| < k$, then with probability at least $3/4$, the algorithm
  outputs a vertex-cut of size at most $k+\gap$ (otherwise it outputs $\bot$).  
\item Otherwise (i.e., no such separation triple $(L',S',R')$ exists), the algorithm outputs either
  a vertex-cut of size at most $k+\gap$ or $\bot$. 
\end{itemize}
\end{theorem}

We obtain exact and $(1+\epsilon)$-approximate local algorithms for
\Cref{thm:localVC_gap} by setting $\gap = 0$ and $\gap =  \lfloor
\epsilon k \rfloor,$ respectively. 

\begin{corollary}
\label{cor:localVC_exact}
There exists the following randomized algorithm. It takes as the same inputs as
in \Cref{thm:localVC_gap} where $\gap$ is set to be equal to $0$.   
It accesses (i.e., makes queries for) $O(\nu k)$ edges and
runs in $O(\nu k^2 )$ time.  It then outputs 
in the following manner.  
\begin{itemize}[noitemsep,nolistsep]
\item  If there exists a separation triple $(L',S',R')$ such that $L'
  \ni x, \vol^{\out}(L') \leq \nu, $ and $|S'| < k$, then with probability at least $3/4$, the algorithm
  outputs a vertex-cut of size at most $k$ (otherwise it outputs $\bot$).  
\item Otherwise (i.e., no such $(L',S',R')$ exists), the algorithm
  outputs either  a vertex-cut of size at most $ k $  or $\bot$. 
\end{itemize}
\end{corollary}

\begin{corollary}
\label{cor:localVC_approx}
There exists the following randomized algorithm. It takes as the same inputs as
in \Cref{thm:localVC_gap} with additional parameter $\epsilon
\in (0,1]$ where $\gap$ is set to be equal to $ \lfloor \epsilon k \rfloor$.  
It accesses (i.e., makes queries for) $O(\nu /\epsilon)$ edges and
runs in $O(\nu k/ \epsilon)$ time.  It then outputs 
in the following manner.  
\begin{itemize}[noitemsep,nolistsep]
\item  If there exists a separation triple $(L',S',R')$ such that $L' \ni x, \vol^{\out}(L')
  \leq \nu, $ and $|S'| <  k$, then with probability at least $3/4$, the algorithm
  outputs a vertex-cut of size at most $\lfloor (1+\epsilon) k \rfloor$ (otherwise it outputs $\bot$).  
\item Otherwise (i.e., no such $(L',S',R')$ exists), the algorithm
  outputs either  a vertex-cut of size at most $\lfloor (1+\epsilon) k
  \rfloor$  or $\bot$. 
\end{itemize}
\end{corollary}
\begin{proof}
The results follow from \Cref{thm:localVC_gap} where we set
$\gap = \lfloor \epsilon k \rfloor$, and \Cref{eq:gap-to-approx}.
\end{proof}

To prove \Cref{thm:localVC_gap}, in \Cref{sec:reduc} we first
reduce the problem to the edge version of the problem using a well-known
reduction (e.g., \cite{Even75,HenzingerRG00}) and then in \Cref{sec:reduc_apply}
we plug the algorithm from \Cref{thm:localEC_approx_new} into the reduction.

\subsection{Reducing from Vertex to Edge Connectivity}

\label{sec:reduc}

Given a directed $n$-vertex $m$-edge graph $G=(V,E)$ and a vertex
$x\in V$, let $G'=(V',E')$ be an $n'$-vertex $m'$-edge graph defined
as follows. We call $G'$ the \emph{split graph} w.r.t. $x$. For each $v\in V-\{x\}$,
we add vertices $v_{\inn}$ and $v_{\out}$ to $V'$ and a directed
edge $(v_{\inn},v_{\out})$ to $E'$. We also add $x$
to $V'$ and we denote $x_{\inn}=x_{\out}=x$ for convenience. For
each edge $(u,v)\in E$, we add $(u_{\out},v_{\inn})$ to $E'$. Suppose,
for convenience, that the minimum out-degree of vertices in $G$ is
$1$. The following two lemmas draw connections between $G$ and $G'$.

\begin{lemma}
\label{lem:completeness}Let $(L,S,R)$ be a separation triple in
$G$ where $S=N_{G}^{\out}(L)$. Let $L'=\{v_{\inn},v_{\out}\mid v\in L\}\cup\{v_{\inn}\mid v\in S\}$
be a set of vertices in $G'$. Then $|E_{G'}(L',V'-L')|=|S|$ and
$\vol_{G}^{\out}(L)\le\vol_{G'}^{\out}(L')\le 3\vol_{G}^{\out}(L)$.
\end{lemma}

\begin{proof}
As $S=N_{G}^{\out}(L)$, we have $|E_{G'}(L',V'-L')|=|\{(v_{\inn},v_{\out})\mid v\in S\}|=|S|$.
Also, $\vol_{G'}^{\out}(L')=\vol_{G}^{\out}(L)+|L|+|S|\le 3\vol_{G}^{\out}(L)$
because every vertex in $G$ has out-degree at least $1$ and $S=N_{G}^{\out}(L)$.
\end{proof}
\begin{figure}
	\centering

\begin{tikzpicture}
	\tikzset{default_node/.style={circle,draw,fill=black,inner sep=0pt,minimum size=4pt}}

	\draw [rounded corners,thick,dashed,green!30!black,fill=green!20] (1.6,0) rectangle (4.4,7.25);
	\draw [rounded corners,thick,dashed,blue!30!black,fill=blue!20] (-1.4,0) rectangle (1.4,4);
	\draw [rounded corners,thick] (-3,-0.25) rectangle (3,4.25);
	\draw [thick] (-3.25,-0.5) rectangle (3.25,6.25);

	\node [anchor=south,blue!30!black] (L) at (-1, 3.25) {\Large $L$};
	\node [anchor=south] (L') at (-2.5, 3.25) {\Large $L'$};
	\node [anchor=south] (L0') at (-2.5, 5) {\Large $L_0'$};
	\node [anchor=south,green!30!black] (S) at (3, 6.5) {\Large $S$};

	\node [default_node,label=below:$x$] (pin) at (0,0.5) {};

	\node [default_node,label=left:$p_\text{in}$] (pin) at (-0.5,3) {};
	\node [default_node,label=left:$q_\text{in}$] (qin) at (-0.5,2) {};
	\node [default_node,label=left:$r_\text{in}$] (rin) at (-0.5,1) {};
	\node [default_node,label=right:$p_\text{out}$] (pout) at (0.5,3) {};
	\node [default_node,label=right:$q_\text{out}$] (qout) at (0.5,2) {};
	\node [default_node,label=right:$r_\text{out}$] (rout) at (0.5,1) {};

	\node [default_node,label=left:$a_\text{in}$] (ain) at (-3.5,3) {};
	\node [default_node,label=left:$b_\text{in}$] (bin) at (-3.5,2) {};
	\node [default_node,label=left:$c_\text{in}$] (cin) at (-3.5,1) {};
	\node [default_node,label=right:$a_\text{out}$] (aout) at (-2.5,3) {};
	\node [default_node,label=right:$b_\text{out}$] (bout) at (-2.5,2) {};
	\node [default_node,label=right:$c_\text{out}$] (cout) at (-2.5,1) {};

	\node [default_node,label=left:$w_\text{in}$] (win) at (2.5,3) {};
	\node [default_node,label=left:$y_\text{in}$] (yin) at (2.5,2) {};
	\node [default_node,label=left:$z_\text{in}$] (zin) at (2.5,1) {};
	\node [default_node,label=right:$w_\text{out}$] (wout) at (3.5,3) {};
	\node [default_node,label=right:$y_\text{out}$] (yout) at (3.5,2) {};
	\node [default_node,label=right:$z_\text{out}$] (zout) at (3.5,1) {};

	\node [default_node,label=left:$u_\text{in}$] (uin) at (2.5,5.75) {};
	\node [default_node,label=left:$v_\text{in}$] (vin) at (2.5,4.75) {};
	\node [default_node,label=right:$u_\text{out}$] (uout) at (3.5,5.75) {};
	\node [default_node,label=right:$v_\text{out}$] (vout) at (3.5,4.75) {};

	\node (anon) at (0.7,3.4) {};

	\draw [->,semithick] (pin) -- (pout);
	\draw [->,semithick] (qin) -- (qout);
	\draw [->,semithick] (rin) -- (rout);

	\draw [->,semithick] (ain) -- (aout);
	\draw [->,semithick] (bin) -- (bout);
	\draw [->,semithick] (cin) -- (cout);

	\draw [->,semithick] (win) -- (wout);
	\draw [->,semithick] (yin) -- (yout);
	\draw [->,semithick] (zin) -- (zout);

	\draw [->,semithick] (uin) -- (uout);
	\draw [->,semithick] (vin) -- (vout);

	\draw [->,semithick] (anon) -- (uin);
	\draw [->,semithick] (anon) -- (vin);
\end{tikzpicture}

	\caption{Construction of a separation triple $(L,S,R)$ in $G$ from an edge cut $L'$
		in $G'$. Most edges are omitted. For this example, $L'_{0}=L'\cup\{u_{\protect\inn},v_{\protect\inn}$\},
		$L=\{r,p,q,x\}$, $S=\{u,v,w,y,z\}$.\label{fig:reduc convert back}}
\end{figure}
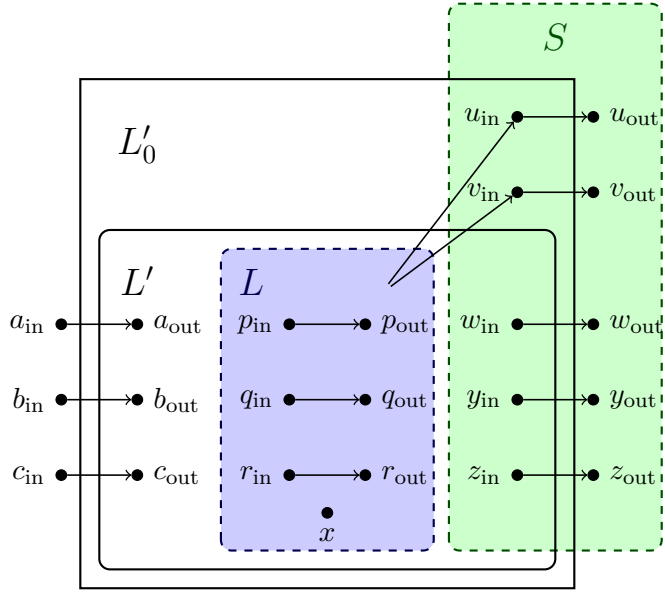

\begin{lemma}
\label{lem:soundness}Let $L'\ni x$ be a set of vertices of $G'$.
Then, there is a set of vertices $L$ in $G$ such that $\vol_{G}^{\out}(L)\le2\vol_{G'}^{\out}(L')$
and $|S|\le|E_{G'}(L',V'-L')|$ where $S=N_{G}(L)$. Given $L'$,
$L$ can be constructed in $O(\vol_{G'}^{\out}(L'))$ time.
Furthermore, for $R=V-(L\cup S)$ we have $R\neq\emptyset$, i.e., $(L,S,R)$ is
a separation triple, if $\vol_{G'}^{\out}(L')\le m'/32$ and $|E_{G'}(L',V'-L')|\le n/2$.

\end{lemma}

\begin{proof}
First, note that if there is $v \in V$ such that $v_{\out}\in L'$ and
$\deg_{G'}^{\out}(v_{\out})\le|E_{G'}(L',V'-L')|$, then we can return
$L=\{v\}$ and $S=N_{G}^{\out}(\{v\})$ and we are done. So from now,
we assume that $\deg_{G'}^{\out}(v_{\out})>|E_{G'}(L',V'-L')|$.

By the structure of $G'$, observe that there are sets $S_{1},S_{2}\subseteq V$
such that 
\[
E_{G'}(L',V'-L')=\{(v_{\inn},v_{\out})\mid v\in S_{1}\}\cup\{(u_{\out},v_{\inn})\mid v\in S_{2}\}
\]
Let $L'_{0}=L'\cup\{v_{\inn}\mid v\in S_{2}\}$. See \Cref{fig:reduc
	convert back} for illustration. So there is a set $S\subsetneq V$ where
\[
E_{G'}(L'_{0},V'-L'_{0})=\{(v_{\inn},v_{\out})\mid v\in S\}.
\]
We have $|S|=|E_{G'}(L'_{0},V'-L'_{0})|\le|E_{G'}(L',V'-L')|$ because
for each $v\in S_{2}$, $\deg^{\out}(v_{\inn})=1\le\deg^{\inn}(v_{\inn})$.
Also, $\vol_{G'}^{\out}(L'_{0})\le\vol_{G'}^{\out}(L')+|E_{G'}(L',V'-L')|\le2\vol_{G'}^{\out}(L')$.

Let $L=\{v\mid v_{\inn},v_{\out}\in L'_{0}\}$. Note that $L\cap S=\emptyset$.
See \Cref{fig:reduc convert back} for illustration. Observe that
$x\in L$ because $x_{\out}=x_{\inn}$. Moreover, $N_{G}(L)=S$ and
$\vol_{G}^{\out}(L)\le\vol_{G'}^{\out}(L'_{0})\le2\vol_{G'}^{\out}(L')$.
$L$ be can constructed in time $O(|\{v_{\out}\in L'\}|)=O(\vol_{G'}^{\out}(L'))$
because the minimum out-degree of vertices in $G$ is $1$.

For the second statement, observe that $R=V-(L\cup S)=\{v\mid v_{\inn}\notin L'_{0}\}$.
Let $V'_{\inn}=\{v_{\inn}\in V'\}\cup\{x\}$ and $V'_{\out}=\{v_{\out}\in V'\}-\{x\}$.
Let $k'=|E_{G'}(L'_{0},V'-L'_{0})|$. Suppose for contradiction that
$R=\emptyset$. We first claim that $R=\emptyset$ implies
\[
|V'-L'_{0}|=|V'_{\out}-L'_{0}|=k'.
\]
This holds because $L'_{0}\supseteq V'_{\inn}$, $V'-L'_{0}\subseteq V'_{\out}$,
and $E_{G'}(L'_{0},V'-L'_{0})$ only contains edges of the form $(v_{\inn},v_{\out})$.
Now, there are two cases. If $m'\ge4n'k'$, then we have 
\begin{align*}
m' & =\vol_{G'}^{\out}(L'_{0})+\vol_{G'}^{\out}(V'-L'_{0})\\
& \le2\vol_{G'}^{\out}(L')+|V'_{\out}-L'_{0}|n'\\
& \le2\cdot m'/32+k'n'\\
& \le m'/16+m'/4<m'
\end{align*}
which is a contradiction. Otherwise, we have $m'<4n'k'$. Note that
$n'<2n$ by the construction of $G'$ and so $m'<8nk'$. Hence, as $\deg_{G'}^{\out}(v_{\out})\le|E_{G'}(L',V'-L')| \geq k $ for every $v_{\out}\in L'$ we
have 
\[
\vol_{G'}^{\out}(L'_{0})\ge|L'_{0}\cap V'_{\out}|k'\ge(n-k')k'\ge nk'/2>m'/16
\]
which contradicts $\vol_{G'}^{\out}(L'_{0})\le2\vol_{G'}^{\out}(L')\le2\cdot m'/32 = m'/16$. 
\end{proof}

\subsection{Proof of \Cref{thm:localVC_gap}}

\label{sec:reduc_apply}
Given an $n$-vertex $m$-edge $G=(V,E)$
represented as adjacency lists, a vertex $x\in V$ and parameters
$\nu,k,\gap$ from \Cref{thm:localVC_gap} where $\nu\le (\gap+1) m/(
12480k)$
and $k\le n/4$, we will work on the split graph $G'$ with $n'$-vertices
$m'$-edges as described in \Cref{sec:reduc}. The adjacency list
of $G'$ can be created ``on the fly''. Let $\localEC(x',\nu',k',\gap')$
denote the algorithm from \Cref{thm:localEC_approx_new} with parameters
$x',\nu',k',\gap'$. We run $\localEC(x,3\nu,k,\gap)$ on
$G'=(V',E')$ in time $O(\nu k^2/(\gap+1))$. Note that $3\nu\le
(\gap+1) m/(8k)\le (\gap+1) m'/(8k)$
as required by \Cref{thm:localEC_approx_new}.

We show that if there exists a separation triple $(L,S,R)$ in $G$
where $L \ni x, |S| < k, $ and $ \vol^{\out}_G(L) \leq \nu$, then
$\localEC(x,3\nu,k,\gap)$ outputs $\bot$ with probability at most
$1/4$.  By \Cref{lem:completeness}, there exists $L'$ in $G'$ such
that $L' \ni x, |E_{G'}(L',V-L')| < k, $ and $ \vol^{\out}_{G'}(L') \leq 3
\vol^{\out}_G(L) \leq 3\nu$. Therefore, by \Cref{thm:localEC_approx_new},
$\localEC(x, 3\nu,k,\gap)$ returns $\bot$ with probability at most
$1/4$.

If, in $G'$, $\localEC(x,3\nu,k,\gap)$ returns $L'$, then by
\Cref{thm:localEC_approx_new} we have  $L'\ni x,
|E_{G'}(L',V'-L')|< k + \gap$ and
$\vol_{G'}^{\out}(L')\le 390\nu k /\gap$.  It remains to show that  we can construct $L
\subsetneq V$ in~$G$ such that $N^{\out}_G(L)$ is a vertex-cut, and
$|N^{\out}_G(L)| < k+\gap$. By \Cref{lem:soundness},
we can obtain in $O(\nu k/(\gap+1))$ time and two sets $L$ and $S=N^{\out}(L)$
where $|S|< k+\gap$. Let $R=V-L\cup S$. As $\vol_{G'}^{\out}(L')\le
390\nu k/(\gap + 1) \stackrel{(\ref{eq:precon_vc})} \le m'/32$
and $k+\gap \le 2k \leq n/2$, we have that
$(L,S,R)$ is a separation triple by \Cref{lem:soundness}. That is, $S=
N^{\out}_G(L)$ is a vertex cut. %

\section{Vertex Connectivity}\label{sec:vertex_con_global}

In this section, we give the first near-linear time algorithm for
checking $k$-vertex connectivity for any $k=\tilde{O}(1)$ in both undirected
and directed graphs.    
\begin{theorem}
\label{thm:VC_undir}There is a randomized (Monte Carlo) algorithm
that takes as input an undirected graph~$G$, a cut-size parameter
$k$, and an accuracy parameter $\epsilon\in(0,1]$, and in time $\tilde{O}(m+nk^{2}/\epsilon)$
either outputs a vertex cut of size less than $\left\lfloor (1+\epsilon)k\right\rfloor $
or declares that $G$ is $k$-vertex connected w.h.p. By setting $\epsilon<1/k$,
the same algorithm decides (exact) $k$-vertex connectivity of $G$ in $\tilde{O}(m+nk^{3})$
time.
\end{theorem}

Now we present the new results for directed graph. 
\begin{theorem}
\label{thm:VC_dir}There is a randomized (Monte Carlo) algorithm that
takes as input a directed graph~$G$, a cut-size parameter $k$,
and an accuracy parameter $\epsilon\in(0,1]$, and in time \\
$\tilde{O}(\min\{mk/\epsilon ,  \poly(1/\epsilon) n^{2+o(1)}\sqrt{k}\})$
either outputs a vertex cut of size less than $\left\lfloor (1+\epsilon)k\right\rfloor $
or declares that $G$ is $k$-vertex connected w.h.p.  For exact vertex connectivity,  there is a randomized  (Monte Carlo) algorithm for exact $k$-vertex connectivity of $G$ in $\tilde{O}(  \min\{ mk^2, k^3n + k^{3/2}m^{1/2}n \} )$ time.
\end{theorem}

To prove \Cref{thm:VC_undir} and \Cref{thm:VC_dir}, we will apply
our framework \cite{NanongkaiSY19} for reducing the vertex connectivity
problem to the local vertex connectivity problem. To describe the
reduction, let $T_{\text{pair}}(m,n,k,\epsilon,p)$ be the time required to either find, for given vertices $s$ and $t$,
an $st$-vertex cut of size less than $\left\lfloor (1+\epsilon)k\right\rfloor $ or to declare
that $s$ and $t$ are $k$-vertex connected correctly with probability at least $1-p$.
Let $T_{\text{local}}(\nu,k,\epsilon,p)$ be the time for solving
correctly with probability at least $1-p$ the local vertex connectivity
problem from \Cref{cor:localVC_approx} when a volume parameter is
$\nu$, the cut-size parameter is $k$, and the accuracy parameter
is $\epsilon$.
\begin{lemma}
[\cite{NanongkaiSY19} Lemma 5.14, 5.15]\label{lem:VC_framework} There is a randomized
(Monte Carlo) algorithm that takes as input a graph $G$, a cut-size
parameter $k$, and an accuracy parameter $\epsilon>0$, and runs
in time proportional to one of these expressions
\begin{align} 
\tilde{O}(m/\overline{\nu})\cdot(T_{\textpair}(m,n,k,\epsilon,1/\poly(n))&+T_{\textlocal}(\overline{\nu},k,\epsilon,1/\poly(n))) \label{eq:edge-sampling} \\ 
\ot (n/\overline{\sigma})\cdot(T_{\textpair}(m,n,k,\epsilon,1/\poly(n)) &+  T_{\textlocal}(\overline{\sigma}^2 + \overline{\sigma}k,k,\epsilon,1/\poly(n))) \label{eq:node-sampling} %
\end{align} 

where $\overline{\nu}\le m,$ and $ \overline{\sigma}\le n$ are parameters that can be
chosen arbitrarily, and either outputs a vertex cut of size less than $\left\lfloor (1+\epsilon)k\right\rfloor $
or declares that $G$ is $k$-vertex connected w.h.p.
\end{lemma}

For completeness, we give a simple proof sketch of \Cref{eq:edge-sampling} which is used
for our algorithm for undirected graphs. The idea for other equations
is similar and also simple.
\begin{proof}[Proof sketch]
	Suppose that $G$ is not $k$-vertex connected. It suffices to give
	an algorithm that outputs a vertex cut of size less than $\left\lfloor (1+\epsilon)k\right\rfloor $
	w.h.p. By considering both $G$ and its reverse graph (where the direction
	of each edge is reversed), there exists w.l.o.g.~a separation triple
	$(L,S,R)$ where $\vol^{\out}(L)\le\vol^{\out}(R)$. There are two
	cases.
	
	Suppose $\vol^{\out}(L)\ge\overline{\nu}$. By sampling $\tilde{O}(m/\overline{\nu})$
	pairs of edges $e=(x,x')$ and $f=(y,y')$, there exists w.h.p.~a pair
	$(e,f)$ where $x\in L$ and $y\in R$. For such a pair $(x,y)$, if
	we check whether $x$ and~$y$ are $k$-vertex connected in time $T_{\text{pair}}(m,n,k,\epsilon,1/\text{\ensuremath{\poly}} (n))$,
	we must obtain an $xy$-vertex cut of size less than $\left\lfloor (1+\epsilon)k\right\rfloor $.
	So, if we check this for each pair $(x,y)$, then we will obtain the
	cut w.h.p. 
	
	Suppose $\vol^{\out}(L)\le\overline{\nu}$. Suppose further that $\vol^{\out}(L)\in(2^{i-1},2^{i}]$. 
	By sampling $\tilde{O}(m/2^{i})$ pairs of edges $e=(x,x')$,
	there exists w.h.p.\ an edge $e$ where $x\in L$. For such vertex $x$,
	if we check the local vertex connectivity in time $T_{\text{local}}(2^{i},k,\epsilon,1/\text{\ensuremath{\poly}} (n))$,
	then the algorithm must return a vertex cut of size less than $\left\lfloor (1+\epsilon)k\right\rfloor $.
	So, if we check this for each pair $(x,y)$, then we will obtain the
	cut w.h.p.
	
	To conclude, the running time in the first case is $\tilde{O}(m/\overline{\nu})\cdot T_{\text{pair}}(m,n,k,\epsilon,1/\text{\ensuremath{\poly}} (n))$. For the second case, we try all $O(\log n)$ many $2^i$, and each such try takes time $\tilde{O}(m/2^{i})\cdot T_{\text{local}}(2^{i},k,\epsilon,1/\text{\ensuremath{\poly}} (n))=\tilde{O}(m/\overline{\nu})\cdot T_{\text{local}}(\overline{\nu},k,\epsilon,1/\text{\ensuremath{\poly}} (n))$
	(if $T_{\text{local}}(\nu,k,\epsilon,1/\text{\ensuremath{\poly}} (n)) = \Omega (\nu) $). This completes the proof of the running time.
	For the correctness, if $G$ is not $k$-vertex connected, we must obtain a desired vertex cut of size $\left\lfloor (1+\epsilon)k\right\rfloor $ w.h.p. So if we do not find any cut, we declare that $G$ is $k$-vertex connected w.h.p.
\end{proof}

\subsection{Undirected Graphs}

Here, we prove \Cref{thm:VC_undir}. First, it suffices to
provide an algorithm with $\tilde{O}(mk/\epsilon)$ time. Indeed, by using the
sparsification algorithm by Nagamochi and Ibaraki \cite{NagamochiI92},
we can sparsify an undirected graph in linear time so that $m=O(nk)$
and $k$-vertex connectivity is preserved. By this preprocessing, the total
running time is $O(m)+\tilde{O}((nk)k/\epsilon))=\tilde{O}(m+nk^{2}/\epsilon)$
as desired. Next, we assume that $k\le\min\{n/4,5\delta\}$ where
$\delta$ is the minimum out-degree of $G$. If $k>5\delta$, then
it is $G$ is clearly not $k$-vertex connected and the out-neighborhood of the
vertex with minimum out-degree is a vertex cut of size less than $k$.
If $k>n/4$, then we can invoke the algorithm by Henzinger, Rao and
Gabow \cite{HenzingerRG00} for solving the problem exactly in time
$O(mn)=O(mk)$.

Now, we have $T_{\text{pair}}(m,n,k,\epsilon,p)=O(mk)$ by the Ford-Fulkerson
algorithm. By repeating the algorithm from \Cref{cor:localVC_approx} $O(\log\frac{1}{p})$
times for boosting its success probability, $T_{\text{local}}(\nu,k,\epsilon,p)=O(\nu k\epsilon^{-1}\log\frac{1}{p})$.
We choose $\overline{\nu}=O(\epsilon m)$ as required by \Cref{cor:localVC_approx}
and also $k\le\min\{n/4,5\delta\}$. Applying \Cref{lem:VC_framework} (\Cref{eq:edge-sampling}),
we obtain an algorithm for \Cref{thm:VC_undir} with running time
\[
\tilde{O}(m/\epsilon m)\cdot O(mk+(\epsilon m)k\epsilon^{-1}\log n)=\tilde{O}(mk/\epsilon).
\]

\subsection{Directed Graphs}
Here, we prove \Cref{thm:VC_dir}.
We again assume that $k\le\min\{n/4,5\delta\}$ using the same reasoning as in the undirected case.
We first show how to obtain the claimed time bound for the approximate problem.
Note that the $\tilde{O}(mk/\epsilon)$-time algorithm follows by the same argument as in the undirected case, 
because both the Ford-Fulkerson algorithm and the local algorithm from \Cref{cor:localVC_approx} work as well in directed graphs.

Next, we give an approximation algorithm with running time $\tilde{O}(\poly(1/\epsilon) n^{2+o(1)}\sqrt{k})$. 
We assume $k \le n^{2/3}$  (for $k \ge n^{2/3}$, we use state-of-the-art $\ot( \poly(1/\epsilon) n^{3+o(1)}/k )$-time algorithm by \cite{NanongkaiSY19}). We have $T_{\text{local}}(\nu,k,\epsilon,p)=O(\nu k\epsilon^{-1}\log\frac{1}{p})$ by \Cref{cor:localVC_approx} and
$T_{\text{pair}}(m,n,k,\epsilon,1/\poly (n))= \ot(\poly(1/\epsilon) n^{2+o(1)})$ using the recent result for $(1+\eps)$-approximating the minimum $st$-vertex cut by Chuzhoy and Khanna \cite{ChuzhoyK19}. By choosing $\overline{\sigma} = n/\sqrt{k}$ for \Cref{lem:VC_framework} (\Cref{eq:node-sampling}), we obtain an algorithm with running time
\begin{align*}
\tilde{O}(n/\overline{\sigma})\cdot(n^{2+o(1)}\poly(1/\epsilon)+(\overline{\sigma}^{2}k+\overline{\sigma}k^{2})/\epsilon) 
& =\tilde{O}(\sqrt{k}\poly(1/\epsilon))\cdot(n^{2+o(1)}+n^{2}+nk^{1.5})\\
& =\tilde{O}(n^{2+o(1)}\sqrt{k}\poly(1/\epsilon)).
\end{align*}

Next, we show how to obtain the time bound for the exact problem. First, observe that we can obtain a $\ot(mk^2)$-time exact algorithm from the $\ot(mk/\epsilon)$-time approximate algorithm by setting $\epsilon < 1/k$. %
It remains to provide an algorithm with the running time $\ot(k^3n + k^{3/2}m^{1/2}n)$.
By \Cref{cor:localVC_exact}, there is an exact algorithm for local vertex connectivity with running time $T_{\text{local}}(\nu,k,1/2k,p)=O(\nu k^2 \log\frac{1}{p})$. 
Also, we have $T_{\text{pair}}(m,n,k,\epsilon,p)= O(mk)$ by the Ford-Fulkerson algorithm. By choosing $\overline{\sigma}  = O(\sqrt{m/k})$ in \Cref{lem:VC_framework} (\Cref{eq:node-sampling}), we obtain an algorithm with running time 
\begin{align*}
\ensuremath{\ot(n/\overline{\sigma})\cdot(mk+(\overline{\sigma}^{2}+\overline{\sigma}k)k^{2})} & =\ot(n/\overline{\sigma})\cdot(mk+(m/k+\sqrt{mk})k^{2})\\
& =\ot(n\sqrt{k/m})\cdot(mk+k^{2.5}\sqrt{m})\\
& =\ot(k^{3/2}m^{1/2}n+k^{3}n).
\end{align*}
Note that $\overline{\sigma}^{2}+\overline{\sigma}k = O(m/k)$ as required by \Cref{cor:localVC_exact}.

	\section{Property Testing} \label{sec:property-testing}

In this section, we give property testing algorithms for
distinguishing between a graph that is $k$-edge/$k$-vertex connected and a graph that is 
$\epsilon$-far from having such property with correct probability at
least $2/3$ for both unbounded-degree and bounded-degree incident-list model.
Recall that for any $\epsilon > 0$, a directed
graph $G$ is $\epsilon$-far from having a property $P$ if at least
$\epsilon m$ edge modifications are needed to make $G$ satisfy
property $P$.  We assume that $\davg = m/n$ is known to the
algorithm at the beginning. %
We state our main results in this section. 

\begin{theorem} \label{thm:main-property-testing}
For the unbounded-degree model, there is a one-sided property testing algorithm
for $k$-edge ($k$-vertex) connectivity  where
$k < \ot (\epsilon n)$ with false reject probability at most $1/3$ that uses
$\ot(k^2/(\epsilon^2 \davg))$ queries (same for $k$-vertex).  If $\davg$ is unknown, then there is
a similar algorithm that uses $\ot(k/\epsilon^2)$ queries (same for
$k$-vertex).   If $G$ is simple, then
the same algorithm for testing $k$-edge connectivity queries at most
  $\ot(\min\{ k^2/(\davg\epsilon^2), k/(\davg\epsilon^3) \})$  (or $\ot(\min\{k/\epsilon^2, 1/\epsilon^3 \})$ edges if $\davg$ is unknown).
\end{theorem}

In the bounded-degree model, we assume that $d$ is known in the beginning. 
\begin{theorem} \label{thm:main-property-testing2} 
For the bounded-degree model, there is a one-sided property testing algorithm
for $k$-edge ($k$-vertex) connectivity where $k < \ot (\epsilon n)$
with false reject probability at most $1/3$ that uses
$\ot(k/\epsilon)$ queries (same for $k$-vertex).  If $G$ is simple, then
the same algorithm for testing $k$-edge connectivity queries at most
 $\ot(\min\{ k/\epsilon, 1/\epsilon^2 \})$. 
\end{theorem}

 We prove \Cref{thm:main-property-testing} using properties of
 $\epsilon$-far from being $k$-edge/vertex connected from
 \cite{OrensteinR11}  and \cite{FrankJ99} along with a variant of
approximate local edge connectivity  in \Cref{sec:test-k-edge}, and
approximate local vertex connectivity in \Cref{sec:test-k-vertex}. 

\subsection{Testing $k$-Edge Connectivity: Unbounded-Degree Model} \label{sec:test-k-edge}
In this section, we prove \Cref{thm:main-property-testing} for testing
$k$-edge connectivity where 
\begin{align} \label{eq:input-k-testing-edge} 
k < \frac{\epsilon n}{520} (\lfloor \log_2(m/n) \rfloor +1) = \ot(\epsilon n ).
\end{align}

 The key tool for our property testing
algorithm is an algorithm for approximate local edge connectivity in a suitable form
for the application to property testing.

\begin{corollary}
\label{lem:gaplocalec}
There exists the following randomized algorithm. It takes as inputs, 
\begin{itemize}[noitemsep,nolistsep]
\item a pointer to a seed vertex $x \in V$ in an adjacency list
  representing  an $n$-vertex $m$-edge graph $G=(V,E)$, 
\item a volume parameter (positive integer) $\nu$, 
\item a cut-size parameter (positive integer) $k$, and 
\item a slack parameter (non-negative integer) $\gap$, where
\end{itemize}
\begin{align} \label{eq:precon_ec_testing}
k \geq 1 + \gap,  \quad  \gap  \leq  k/2 \quad \mbox{ and } \quad \nu < m(\gap+1)/(130(k-\gap)).
\end{align}
It accesses (i.e., makes queries for) $\ot(\nu k /(\gap+1))$ edges. It then outputs 
in the following manner. 
\begin{itemize}[noitemsep,nolistsep]
\item If there exists a vertex-set $S'$ such that $S' \ni x, \vol^{\out}(S')
  \leq \nu, $ and $ |E(S',V-S')| < k-\gap$, then with probability at
  least 3/4, the algorithm
  outputs $S$ a non-empty vertex-set $S\subsetneq V$
  such that $|E(S,V-S)| < k$  (otherwise it outputs $\bot$).  
\item Otherwise (i.e., no such $S'$ exists), the algorithm outputs either
  a set $S$ as specified above or $\bot$. 
\end{itemize}
\end{corollary}
\begin{proof}
If $\nu < k - \gap$, then it is enough to run the following
procedure. If $\deg(x) < k - \gap$, then we found a cut of size
smaller than $k$, and return $\{ x\}$. Otherwise, we return
$\bot$. To see that we return $\bot$ correctly, for any
vertex set $S'$ that contains $x$, we have $\vol^{\textout}(S') \geq
k -\gap > \nu$, and thus we correctly output $\bot$.   
From now, we assume that 
\begin{align}  \label{eq:nu-geq-k-gap}
\nu \geq k -\gap.
\end{align}
Let the following tuple $(x',\nu', k', \gap')$ be a list of parameter
supplied to \Cref{thm:localEC_approx_new}. We run the local algorithm in 
\Cref{thm:localEC_approx_new}  using the following parameters
$(x',\nu', k', \gap') = (x, \nu, k-\gap, \gap)$. Note that
\Cref{eq:precon_ec_testing} and \Cref{eq:nu-geq-k-gap}  ensure that the conditions in
\Cref{eq:precon_ec} for the local algorithm with
parameter $(x',\nu', k', \gap')$ are satisfied. The query complexity
follows from \Cref{thm:localEC_approx_new}. 
\end{proof}
We now present an algorithm for testing $k$-edge connectivity.
 
\paragraph{Algorithm.}  
\begin{enumerate}
\item Sample $\Theta (\frac{1}{\epsilon})$ vertices uniformly at random.
\item If any of the sampled vertices has out-degree less than $k$, return the
  corresponding trivial edge cut. 
\item Sample $\Theta( \frac{ k \log k} { \epsilon \davg })$ vertices
  uniformly at random  (if $\davg$ is unknown, then sample $\Theta( \frac{\log
    k}{\epsilon})$ vertices instead).
\item For each sampled vertex $x$, and for $i \in \{0, 1 ,\ldots,
   \logk \}$, 
 \begin{enumerate}
 \item Let $\nu = 2^{i+2} \epsilon^{-1} \lfloor \log_2k \rfloor,$ and
   $\gap = \min\{2^{i}-1, \lfloor k/2 \rfloor \}$. 
 \item Run the local algorithm in \Cref{lem:gaplocalec} with parameters $(x,
   \nu , k, \gap )$ on both $G$ and $G^R$ where  $G^R$ is $G$ with
   reversed edges. 
 \end{enumerate} 
\item Return an edge cut of size less than $ k $ if any execution of
  the local algorithm above returns a cut. Otherwise, declare that $G$ is $k$-edge connected.
\end{enumerate}

\paragraph{Query Complexity.} We first show that the number of
edge queries is at most $\ot(k^2/(\epsilon^2 \davg))$. For each
sampled vertex $x$ and $i \in \{0,1,  \ldots, \logk \}$, by
\Cref{lem:gaplocalec}, the local algorithm queries $\ot( \nu k/\gap) =
\ot(k/\epsilon)$ edges. The result follows as we repeat $\log_2k$
times per sample, and we sample $O(k \log k /(\epsilon \davg))$
times. %
 If $\davg$ is unknown, we can remove the term $k/\davg$ from
 above since we sample $\Theta( \frac{\log  k}{\epsilon})$ vertices instead.   

\paragraph{Correctness.} If $G$ is $k$-edge connected, 
the algorithm above never returns an edge cut. We show that if $G$ is
$\epsilon$-far from being $k$-edge connected, then the algorithm outputs an
edge cut of size less $ k$ with constant probability. We start with a simple
observation. 
\begin{lemma}  \label{lem:mgeqnk}
If $m < nk/4$, then with constant probability, the algorithm outputs
an edge cut of size less than $k$ at step 2. 
\end{lemma}
\paragraph{Remark.} This observation applies to any tester.
\begin{proof}
Suppose $m < nk/4$. There are at most $n/2$ vertices with out-degree at least
$k$. Hence, there are at least $n/2$ vertices of degree
less than $k$. In this case, we can sample $O(1)$ time where each sampled vertex
$x$ we check $\deg^{\textout}(x)  < k$ to get $k$-edge cut with
constant probability.
\end{proof}
From now  we assume that 
\begin{align} \label{eq:mgeqnk4}
m \geq nk/4.
\end{align}

Next, we state important properties when $G$ is $\epsilon$-far from
being $k$-edge connected. For any non-empty subset $X \subset V$, let
$d^{\textout}(X) = |E(X, V - X)|$, and $d^{\textin}(X) = |E(V-X,X)|$.  
\begin{theorem} [\cite{OrensteinR11} Corollary
  8]\label{thm:eps-far-edge}
A directed graph $G = (V,E)$ is $\epsilon$-far from being
$k$-edge connected (for $k \geq 1$) if and only if there exists a
family of disjoint subsets $\{X_1, \ldots, X_t\}$ of vertices for
which either $\sum_{i}(k - d^{\textout}(X_i)) > \epsilon m$ or $
\sum_{i}(k - d^{\textin}(X_i)) > \epsilon m$. 
\end{theorem} %
Let $\mathcal{F} :=  \{X_1, \ldots, X_t\}$ as in
\Cref{thm:eps-far-edge}.  We assume without loss of generality that
\begin{align} \label{eq:eps-far-edge} 
\sum_{i}(k - d^{\textout}(X_i)) > \epsilon m.
\end{align}  Let $\mathcal{C}_{-1} = \{ X \in \mathcal{F} \colon k \leq d^{\textout}(X) \}$. For $i
\in \{0, 1,\ldots, \lfloor \log_2k \rfloor \}$, let  $\mathcal{C}_i = \{ X \in \mathcal{F}
\colon k - d^{\textout}(X) \in [2^i, 2^{i+1}) \}$. Note that 
\begin{align} \label{eq:2ileqk}
   2^i \leq k, \text{ for any } i \in \{0, \ldots, \logk \}
\end{align}and
\begin{align} \label{eq:partitionF-edge}
 \mathcal{F} = \bigsqcup_{i = -1}^{ \lfloor \log_2k \rfloor} \mathcal{C}_i
\end{align}
where $\bigsqcup$ is the disjoint union.  Let $\cC_{i, \textbig} = \{
X \in \cC_i \colon \vol^{\textout}(X) \geq  2^{i+2} \epsilon^{-1} (\logk+1) \}$, and $\cC_{i, \textsmall} = \cC_i - \cC_{i, \textbig}$. 
The following lemma is the key for the algorithm's correctness. 
\begin{lemma} \label{lemma:many-small-edge}
There is some $i^* \in \{0, 1,\ldots, \logk\}$ such that $ |\cC_{i^*,\textsmall}| \geq \epsilon n \davg
/ (4k  (\logk+1)) $. If $\davg$ is unknown, we have
$|\cC_{i^*,\textsmall}| \geq  \epsilon n/(16(\logk+1))$ instead. 
\end{lemma}

Let us briefly argue that \Cref{lemma:many-small-edge} implies the correctness of
the algorithm. If \Cref{lemma:many-small-edge} is true, then  by
sampling $\Theta ((k \log k)/(\epsilon \davg))$ many vertices (or $\Theta (\log k/\epsilon)$ if $\davg$ is unknown), the event that a sampled vertex belongs to some vertex set $S' \in
\cC_{i^*,\textsmall}$ has constant probability (since
$\cC_{i^*,\textsmall}$ contains disjoint sets). Now, assuming that a
sampled vertex $x$ belongs to a vertex set $S'$ as specified above. Since $S' \in
\cC_{i^*,\textsmall}$, by definition, $\vol^{\textout}(S') \leq
2^{i^*+2}\epsilon^{-1}(\lfloor \log_2 k\rfloor + 1)$, and $k -
d^{\out}(S') \in [2^{i^*}, 2^{i^*+1})$. In other words, there exists a
vertex set $S'$ such that $S' \ni x, \vol^{\out}(S') \leq
2^{i^*+2}\epsilon^{-1}(\lfloor \log_2 k\rfloor + 1),$ and
$|E(S',V-S')| \leq k - 2^{i^*} < k - (2^{i^*}-1)$. If $i^*$ was known, then we run the local algorithm in
\Cref{lem:gaplocalec} with parameters $(x,\nu,k,\gap)$ where $\nu =
2^{i^*+2}\epsilon^{-1}(\lfloor \log_2 k\rfloor + 1), $ and $\gap
=\min\{2^{i^*}-1, \lfloor k/2 \rfloor \} $ to obtain a non-empty set $S$ such that
$|E(S,V-S)| < k$ with probability at least $3/4$. Since $i^* \in [0, 1, \ldots, \logk\}$, it is
enough to run the local algorithm in \Cref{lem:gaplocalec} with
parameters $(x,\nu_i,k,\gap_i)$ where $\nu_i = 2^{i+2}\epsilon^{-1}(\logk+1)$ and $\gap_i = \min\{2^{i}-1, \lfloor k/2 \rfloor \}$ for every $i \in
\{0, 1,\ldots, \logk\}$.

It remains to check the preconditions for the local algorithm in
\Cref{lem:gaplocalec}. We show that
\Cref{eq:precon_ec_testing} is satisfied for any call of the local
algorithm. It is easy to see that the first condition $k \geq 1+ \gap$ is
satisfied.  The second condition $\gap \leq k/2$ follows since
$\gap_i = \min\{2^{i}-1, \lfloor k/2 \rfloor \}  \leq \lfloor k/2
\rfloor \leq k/2$ for any $i \in \{0, 1,\ldots, \logk\}$. Finally, we show
the last condition $\nu < m(\gap+1)/(130(k-\gap))$ is
satisfied. First, note that for any $i \in \{0, 1,\ldots, \logk\},
\nu_i = 2^{i+2}\epsilon^{-1}(\logk+1) \leq 4\epsilon^{-1}k(\logk+1)$.
Therefore, 
\begin{align}
\nu \leq  4\epsilon^{-1}k(\logk+1) \stackrel{~(\ref{eq:mgeqnk4})} \leq
  4 \frac{m\epsilon^{-1}}{n} \lfloor \log_2(m/n)+1   \rfloor \stackrel{~(\ref{eq:input-k-testing-edge})} < \frac{m}{130k} < m(\gap+1)/(130(k-\gap))
\end{align}
This last inequality follows since $0 \leq \gap < k/2$.  Therefore, the correctness follows.

\begin{proof}[Proof of \Cref{lemma:many-small-edge}] 
We begin by showing that there is some $i \in \{0, 1,\ldots, \logk \}$ such that $|\cC_i| > \epsilon m /
(2^{i+1} (\logk+1) )$.
that there is some $i \in \{0, 1,\ldots, \logk \}$ such that 
\begin{align}\label{eq:property:one}
\sum_{ X \in \mathcal{C}_i} (k - d^{\textout}(X)) > \epsilon m/
(  \lfloor \log_2k \rfloor+1).
\end{align}
 Suppose otherwise that for every $i \in \{0, 1,\ldots, \logk \}$,
$\sum_{ X \in \mathcal{C}_i } (k - d^{\textout}(X)) \leq \epsilon m/
(\lfloor \log_2k \rfloor+1)$.
Then we have $\sum_{ X \in \mathcal{F}}(k- d^{\textout}(X) )
\stackrel{(\ref{eq:partitionF-edge})} = \sum_{i=-1}^{\lfloor \log_2k
  \rfloor }  \sum_{ X \in \mathcal{C}_i} (k - d^{\textout}(X))  \leq
\sum_{i=0}^{\lfloor \log_2k \rfloor}
\sum_{ X \in \mathcal{C}_i} (k - d^{\textout}(X)) \leq \epsilon m$. However, this
contradicts  \Cref{eq:eps-far-edge} as in
\Cref{thm:eps-far-edge}. Second, we claim that for every
$i \in \{0, 1,\ldots, \logk \}$,  $$|\mathcal{C}_i| 2^{i+1} \geq \sum_{ X \in \mathcal{C}_i} (k - d^{\textout}(X)).$$
This follows trivially from the definition that each element $X$ of the set
$\mathcal{C}_i$ has $k - d^{\textout}(X) < 2^{i+1}$.
Now for the $i$ that satisfies \Cref{eq:property:one} we have 
\begin{align}\label{eq:property:two}
|\mathcal{C}_i|  \geq
\sum_{ X \in \mathcal{C}_i} (k - d^{\textout}(X)) / 2^{i+1} > \epsilon
  m/ ( 2^{i+1} (\lfloor \log_2k \rfloor+1))
\end{align}
as desired.

Recall that $\cC_{i, \textbig} = \{ X \in \cC_i \colon
\vol^{\textout}(X) \geq 2^{i+2} \epsilon^{-1} (\logk+1)  \}$, and $\cC_{i,
  \textsmall} = \cC_i - \cC_{i, \textbig}$. We show that $
|\cC_{i,\textbig}| <  |\cC_i|/2$ for $i$ that satisfies \Cref{eq:property:two} by the following chain of inequalities:
$$    2|C_{i,\textbig}| \leq \sum_{X \in \cC_{i, \textbig}}
\vol^{\textout}(X)/(  2^{i+1} \epsilon^{-1} (\lfloor \log_2k \rfloor+1) )
\leq \epsilon m /(2^{i+1}(\lfloor \log_2k \rfloor+1) )
\stackrel{(\ref{eq:property:two})}  < |\cC_i|. $$
The first inequality holds because by the inequality $\vol^{\textout}(X)/( 2^{i+1}
\epsilon^{-1} (\logk+1)) \geq 2$ for each $X \in
\cC_{i,\textbig}$ from the definition of $\cC_{i,\textbig}$, we have $ \sum_{X \in \cC_{i, \textbig}}
\vol^{\textout}(X)/( 2^{i+1} \epsilon^{-1} (\lfloor \log_2k \rfloor+1) )
\geq 2|C_{i,\textbig}| $. The second inequality holds because elements in~$\cC_{i, \textbig} $
are disjoint and thus $\sum_{X \in \cC_{i, \textbig}} \vol^{\textout}(X)
\leq m$. The final inequality follows from  \Cref{eq:property:two}.

For the same $i$, since  $|\cC_{i,\textbig}| <  |\cC_i|/2$, we
have \begin{align} \label{eq:last-ineq}
|\cC_{i,\textsmall}| \geq |\cC_i|/2 > \epsilon m / (2^{i+2}( \lfloor
       \log_2 k\rfloor+1) ) \geq   \epsilon n \davg / (4k  (\logk+1)). 
\end{align}  %
The last inequality follows from $m = n\davg$, and \Cref{eq:2ileqk}. If $\davg$ is unknown, by \Cref{eq:mgeqnk4}, the last inequality
becomes $\epsilon m / (2^{i+2} (\logk+1)) \geq  \epsilon nk / (16k
\logk) = \epsilon n / (16(\logk+1)) $.
\end{proof}
\paragraph{An improved bound for simple graphs.} The same algorithm gives
an improved bound when $G$ is simple.   If $\epsilon \leq 4/k$, the algorithm queries at most
$\ot(k^2/(\epsilon^2 \davg)) = \ot(1/(\epsilon^4 \davg)) $ edges
(and $\ot(1/\epsilon^3)$ edges if $\davg$ is unknown). Now, we assume
$\epsilon > 4/k$, we show that there are $\Omega(\epsilon n \davg/
k)$ ($\Omega(\epsilon n)$ if $\davg$ is unknown) many vertices with
degree less than $k $.

\begin{lemma} \label{lem:lots-singleton}
 If $\epsilon > 4/k$, $G$ is simple, and  $\epsilon$-far from being
 $k$-edge connected, then there exist at least $\epsilon n/2$ vertices
 ($\epsilon \davg n/(8k)$ vertices if $\davg$ is unknown) with
 degree less than $k$. 
\end{lemma}

\Cref{lem:lots-singleton} immediately yields the correctness of the
algorithm as number of singleton with degree less than $k$ is at least $\epsilon n/2$ vertices
 ($\epsilon \davg n/(8k)$ vertices if $\davg$ is unknown), and we
 sample $\Theta( k/(\epsilon \davg))$ (or $\Theta(1/\epsilon)$ vertices if $\davg$ is unknown) at step 1
and 2 to check if each sampled vertex has degree less than $k$.  Next, we prove \Cref{lem:lots-singleton}. 

\begin{proof} [Proof of \Cref{lem:lots-singleton}]

Let $\cC = \{ X \colon k - d^{\textout}(X)\geq 1 \}$. We claim that 
\begin{align} \label{eq:cC-g-emk} |\cC| > \epsilon m/k. \end{align} This follows from 
$$|\cC|k \geq \sum_{X \in \cC} (k - d^{\textout}(X)) \geq \sum_{X \in
  \mathcal{F}} (k - d^{\textout}(X)) > \epsilon m.$$   The first inequality follows from each term $k -
d^{\textout}(X)$ is at most $k$, and there are $|\cC|$ terms. The
second inequality follows from each $X \in \mathcal{F}
- \cC$, $k - d^{\textout}(X) \leq 0$. The third inequality follows from
\Cref{eq:eps-far-edge}. 

Let $\cC_{\textbig} = \{ X \in \cC \colon \vol^{\textout}(X)\geq
2k/\epsilon \}$, and $\cC_{\textsmall} = \cC - \cC_{\textbig}$.  We
claim that \begin{align} \label{eq:cC-small-g-eps-n}|\cC_{\textsmall}| > \epsilon n/8. \end{align} First, we show that
\begin{align} \label{eq:cC-big-leq-cC} |\cC_{\textbig}| < |\cC|/2. \end{align} This follows from $$ |\cC_{\textbig}|
\leq \sum_{X \in \cC_{\textbig}} \vol^{\textout}(X)/(2k
\epsilon^{-1}) \leq  \epsilon m/(2k) < |\cC|/2.  $$
The first inequality follows from the fact that for each $X \in \cC_{\textbig}, \vol^{\textout}(X)/
(2k \epsilon^{-1}) \geq 1$. Hence,  $\sum_{X \in \cC_{\textbig}} \vol^{\textout}(X)/
(2k \epsilon^{-1}) \geq  |\cC_{\textbig}|$.  The second inequality
follows from the fact that $\cC_{\textbig}$ contains disjoint sets,
and $ \sum_{X \in \cC_{\textbig}} \vol^{\textout}(X) \leq
\vol^{\textout}(V) = m$. The last inequality follows from
\Cref{eq:cC-g-emk}. Next, we have  \begin{align} \label{eq:cctextsmall2k}|\cC_{\textsmall}|
                                     \geq |\cC|/2 \geq \epsilon m /
                                     (2k) \geq \epsilon (n\davg)/
                                     (2k) \geq \epsilon (n\davg)/2. \end{align}
The first inequality follows from \Cref{eq:cC-big-leq-cC} and that  $\cC_{\textsmall} = \cC - \cC_{\textbig}$. The second
inequality follows from \Cref{eq:cC-g-emk}. The third inequality
follows from  $m = n\davg$.  If $\davg$ is unknown, the last part
of \Cref{eq:cctextsmall2k} becomes $ m / (2k) \geq \epsilon (nk )/ (8k) \geq \epsilon
n/8$. This follows from \Cref{eq:mgeqnk4}.%

It suffices to show that, for each $X\in\cC_{\textsmall}$, the average
degree of vertices in $X$, which is $\frac{\vol^{\textout}(X)}{|X|}$, is less
than $k$. If this is true, then there exists vertex $x\in X$ where
$\deg x<k$. Since the sets in~$\cC$ are disjoint, each set $X \in
\cC$ contains a vertex with degree less than $k$, and
$|\cC_{\textsmall}| > \epsilon n/8$ (by \Cref{eq:cC-small-g-eps-n}), we have that the number of
singleton vertex with degree less than $k$ is $> \epsilon n/8$, and we
are done. 

Now, fix $X\in\cC_{\textsmall}$ and we want to show that $\frac{\vol^{\textout}(X)}{|X|}<k$.
Consider three cases. If $|X|=1$, then $\frac{\vol^{\textout}(X)}{|X|}=d^{\textout}(X)<k$.
Next, if $|X|\ge2/\epsilon$, then $\frac{\vol^{\textout}(X)}{|X|}<\frac{2k/\epsilon}{2/\epsilon}=k$
as $X\in\cC_{\textsmall}$. In the last case, we have $2\le|X|<2/\epsilon\le k/2$.
Here we use the assumption $ \epsilon < 4/k $.
Note that $\vol^{\textout}(X)\le d^{\textout}(X)+|X|^{2}$ because the graph is simple.
So, 
\[
\frac{\vol(X)}{|X|}\le\frac{d^{\textout}(X)+|X|^{2}}{|X|}<\frac{k}{|X|}+|X|<\frac{k}{2}+\frac{k}{2}=k.\qedhere
\]
\end{proof}

\subsection{Testing $k$-Edge Connectivity: Bounded-Degree Model} \label{sec:test-k-edge-bounded}
In this section, we prove \Cref{thm:main-property-testing} for testing $k$-edge connectivity for
bounded degree model.  In this model, we know the maximum out-degree
$d$. We assume that $G$ is $d$-regular, meaning that every vertex has
degree $d$. If $G$ is not $d$-regular, we can ``treat'' $G$ as if it
is $d$-regular as follows. For any list $L_v$ of size less than $ d$, and $i
\in (|L_v|, d]$, we ensure that query$(v,i)$ returns a self-loop edge (i.e.,
an edge $(v,v)$). 

\paragraph{Edge-sampling procedure.} The key property of a $d$-regular graph is that we can easily sample edges
uniformly.
The other arguments made in the following would in principle also apply to the unbounded-degree model.
However, in the unbounded-degree model any algorithm that samples an edge from an almost uniform distribution must perform $ \Omega (n/\sqrt{m}) = \Omega (\sqrt{n} / \sqrt{d}) $ queries~\cite{EdenR18}, which is not feasible for our purposes.
The sampling in the bounded-degree model works as follows: We first sample a vertex $x \in V$. Then, we
perform query$(x,i)$ where $i$ is an integer sampled uniformly from
$[1,d]$. 
\begin{proposition}  \label{pro:sample-edge}
For any edge $e \in E$, the probability that $e$ is sampled from the edge-sampling procedure is $1/m$.
\end{proposition}
\begin{proof}
Fix any edge $e \in E$. The edge $e$ belongs to some list
$L_v$. Therefore, the probability that $e$ is queried according to our
edge-sampling procedure is
\begin{align*}
\Pr [e \text{ is queried}] &=   \Pr [e \text{ is queried} \mid L_v \text{ is
                          sampled}] \cdot \Pr [L_v \text{ is sampled}] \, + \\ 
& ~~~\,\Pr [e \text{ is queried} \mid L_v \text{ is
                       not   sampled}] \cdot \Pr [L_v \text{ is not sampled} ] \\
 &= \Pr [e \text{ is queried} \mid L_v \text{ is
                          sampled}] \cdot \Pr [L_v \text{ is sampled} ] \\
&= (1/d) \cdot (1/n) = 1/m.\qedhere
\end{align*}
\end{proof}

We present an algorithm for testing $k$-edge connectivity for
bounded-degree model and analysis. %
 
\paragraph{Algorithm.}  
\begin{enumerate}
\item Sample $\Theta (\frac{1}{\epsilon})$ vertices uniformly at random.
\item If any of the sampled vertices has out-degree less than $k$, return the
  corresponding trivial edge cut. 
\item For each $i \in \{0,\ldots,\logk\}$ and each $j \in
  \{0,\ldots, \logeta\}$ where $\eta_i = 2^{i+2}\epsilon^{-1}\logk$,
\begin{enumerate}
\item Sample $\Theta (\frac{ \logk\logeta }{ \epsilon 2^{j-i} }) = \tilde \Theta( \frac{1}{\epsilon 2^{j-i}}) $  edges uniformly at random.   
\item Let $\nu = 2^{j+1},$ and $\gap =  \min\{2^{i}-1, \lfloor k/2 \rfloor \}$. 
 \item Run the local algorithm in \Cref{lem:gaplocalec} with
   parameters $(x, \nu , k, \gap )$ on both $G$ and $G^R$ where
   $G^R$ is $G$ with reversed edges, and $x$ is a vertex from the
   sampled edge of the form $(x,y)$.  
\end{enumerate}
\item Return an edge cut of size less than $ k $ if any execution of
  the local algorithm above returns a cut. Otherwise, declare that $G$ is $k$-edge connected.
\end{enumerate}

\paragraph{Query Complexity.}  We first show that the number of
edge queries is at most $\ot(k/\epsilon)$. For each
vertex $x$ from the sampled edge $(x,y)$ and
for each $(i,j)$ pair in loops, the local algorithm in \Cref{lem:gaplocalec}
queries $\ot( \nu k/\gap) = \ot(2^{j-i}k)$ edges, and we sample $\ot(
1/(\epsilon 2^{j-i}))$ times.  Therefore, by repeating $\ot(1)$ time, the total edge queries is at most $\ot( k/ \epsilon)$.

\paragraph{Correctness.} If $G$ is $k$-edge connected, then the algorithm
never returns any edge cut, and we are done. Suppose $G$ is
$\epsilon$-far from being $k$-edge connected, then we show that the
algorithm outputs an edge cut of size less than $k$ with constant
probability.  Since $G$ is $d$-regular, we have $\davg =
d$. Therefore,  we can use results from \Cref{sec:test-k-edge}. %
 Let $\mathcal{F} :=  \{X_1, \ldots, X_t\}$ as in
\Cref{thm:eps-far-edge}.  We assume without loss of generality that
\begin{align} \label{eq:eps-far-edge2} 
\sum_{i}(k - d^{\textout}(X_i)) > \epsilon m.
\end{align}  Let $\mathcal{C}_{-1} = \{ X \in \mathcal{F} \colon k \leq d^{\textout}(X) \}$. For $i
\in \{0, 1,\ldots,\logk\}$, let  $\mathcal{C}_i = \{ X \in \mathcal{F}
\colon k - d^{\textout}(X) \in [2^i, 2^{i+1}) \}$.  Let $\cC_{i,
  \textbig} = \{ X \in \cC_i \colon \vol^{\textout}(X) \geq 2^{i+2}
(\logk+1) /\epsilon \}$, and $\cC_{i, \textsmall} = \cC_i - \cC_{i,
  \textbig}$.  By \Cref{lemma:many-small-edge},  there is an $i \in \{0, 1,\ldots, \logk \}$ such
that  \begin{align} \label{eq:ccitextsmallgeqepsilon} |\cC_{i,\textsmall}| \geq \epsilon n \davg
/ (4k  (\logk+1)) = \epsilon m/ (4k (\logk+1)). \end{align} This last inequality
follows since $n\davg = nd = m$. 
Now let $\eta_i = 2^{i+2}\epsilon^{-1}\logk$, and for every $j \in \{0, 1, \ldots, \logeta \}$,
let $\cC_{i,\textsmall, j} = \{ X \in \cC_{i,\textsmall} \colon \vol^{\textout}(X) \in [2^j, 2^{j+1}) \}$. 

\begin{lemma} \label{lem:bigvol-small}
For the $i \in \{0, \ldots, \logk \}$ that satisfies  \Cref{eq:ccitextsmallgeqepsilon}, there is a $j \in \{0, 1, \ldots, \logeta \}$ such that $\sum_{X \in \cC_{i,\textsmall,j}}
\vol^{\textout}(X) \geq \epsilon m 2^{j-i}/ (4 (\logk+1)( \logeta+1))$. 
\end{lemma}

Let us briefly argue that \Cref{lem:bigvol-small} implies the correctness of the algorithm. By
sampling $\Theta (\frac{ \logk\logeta }{ \epsilon 2^{j-i} }) = \tilde
\Theta( \frac{1}{\epsilon 2^{j-i}}) $ edges, we get
that the event that a sampled edge $(u,v)$ has $u \in X$ for some $X \in
\cC_{i,\textsmall, j}$ with constant probability (since
$\cC_{i,\textsmall,j}$ contains disjoint elements).  For each $(i,j)  \in
\{0, 1, \ldots, \logk\} \times \{0,\ldots, \logeta\}$, we run the
local algorithm in \Cref{lem:gaplocalec}  with $\nu = 2^{j+1}$, and
$\gap = \min\{2^{i}-1, \lfloor k/2 \rfloor \}$; also, there exists a pair $(i,j)$ such that $\sum_{X \in \cC_{i,\textsmall,j}}
\vol^{\textout}(X) \geq \epsilon  m 2^{j-i}/ (4 (\logk+1)( \logeta+1)) $ by
\Cref{lem:bigvol-small}. Therefore, the local algorithm in \Cref{lem:gaplocalec}, 
outputs an edge cut of size less than $k$ with constant probability.
Note that the preconditions of the local algorithm in
\Cref{lem:gaplocalec} can be shown to hold in a similar manner as in the previous subsection .

\begin{proof}[Proof of \Cref{lem:bigvol-small}] 
We first claim that there is $j \in \{0, 1, \ldots, \logeta \}$ such that  \begin{align} \label{eq:citextsmallj} |\cC_{i,\textsmall,
                                        j}| \geq |\cC_{i,\textsmall}| /
                                        (\logeta+1)  \end{align}
Suppose otherwise. Then we have that $ |\cC_{i,\textsmall,j}| <
|\cC_{i,\textsmall}|/ (\logeta+1) $ for all $j \in \{0, \ldots, \logeta\} $, and therefore $\sum_{j \in
  \{0,\ldots, \logeta \} } |\cC_{i,\textsmall,j}| <
|\cC_{i,\textsmall}|$, a contradiction.

Now, for the $j \in \{0, 1, \ldots, \logeta \}$ that satisfies~\eqref{eq:citextsmallj} we have 
\begin{align*} 
 \sum_{X \in \cC_{i,\textsmall,j}} \vol^{\textout}(X)& \geq
|\cC_{i,\textsmall,j}|2^j \\&\stackrel{(\ref{eq:citextsmallj})} \geq
|\cC_{i,\textsmall}|2^j/ (\logeta+1) \\&
\stackrel{(\ref{eq:ccitextsmallgeqepsilon})} \geq
\epsilon m 2^j/ (4 k (\logk+1)( \logeta+1) ) \\&
\stackrel{(\ref{eq:2ileqk})} \geq
\epsilon m 2^{j-i}/ (4 (\logk+1) (\logeta+1)) 
\end{align*} 
The first inequality holds because the set $\cC_{i,\textsmall,j}$
contains disjoint elements and $\vol^{\textout}(X) \geq 2^j$ by the
definition of $\cC_{i,\textsmall,j}$. %
\end{proof}

\paragraph{An improved bound for simple graphs.} The same algorithm gives
the improved bound of $\ot(\min\{ k/\epsilon, 1/\epsilon^2 \})$ queries when $G$
is simple.   If $\epsilon \leq 4/k$, then the algorithm queries at
most $\ot(k/\epsilon) = \ot(1/\epsilon^2) $ edges. Otherwise, if $\epsilon >
4/k$, by \Cref{lem:lots-singleton},  there are $\Omega(\epsilon n)$
many vertices with degree less than $k$, and this implies that the
algorithm outputs an edge cut of size less than $k$ at step 2.  

\subsection{Testing $k$-Vertex Connectivity: Unbounded-Degree Model} \label{sec:test-k-vertex}

In this section, we prove \Cref{thm:main-property-testing} for testing
$k$-vertex connectivity where 
\begin{align} \label{eq:input-k-testing-vertex}
k < \frac{n}{8} \min\{1, \frac{\epsilon}{12480} \lfloor \log_2(m/n)+1
  \rfloor \} = \ot(\epsilon n)
\end{align} 
 The key tool for our property testing
algorithm is approximate local vertex connecitvity in a suitable form
for the application to property testing.  

\begin{corollary}
\label{lem:gaplocalvc}
There exists the following randomized algorithm. It takes as inputs, 
\begin{itemize}[noitemsep,nolistsep]
\item a pointer to  an adjacency list representing  an $n$-vertex $m$-edge graph
$G=(V,E)$, 
\item a seed vertex $x \in V$, 
\item a volume parameter (positive integer) $\nu$, 
\item a cut-size parameter (positive integer) $k$, and 
\item a slack parameter (non-negative integer) $\gap$, where
\end{itemize}
\begin{align} \label{eq:precon_vc_testing}
1 \leq k - \gap \leq n/4, \quad  \gap \leq k/2 \quad \mbox{ and } \quad \nu < m(\gap+1)/(12480(k-\gap)).
\end{align}
It accesses (i.e., makes queries for) $\ot(\nu k /(\gap+1))$ edges. 
 It then outputs 
in the following manner. 
\begin{itemize}[noitemsep,nolistsep]
\item If there exists a separation triple $(L',S',R')$ such that $L' \ni x, \vol^{\out}(L')
  \leq \nu, $ and $|S'| < k-\gap$, then   with probability at
  least 3/4, the algorithm
  outputs a vertex-cut of size at most $k$ (otherwise it outputs $\bot$).  
\item Otherwise (i.e., no such separation triple $(L',S',R')$ exists), the algorithm outputs either
  a vertex-cut of size at most $k$ or $\bot$. 
\end{itemize}
\end{corollary}
 \begin{proof}
If $\nu < k - \gap$, then it is enough to run the following
procedure. If $\deg(x) < k - \gap$, then we found a cut of size
smaller than $k$, and return $N^{\out}(x)$. Otherwise, we return
$\bot$. To see that we return $\bot$ correctly, for any
vertex set $L'$ that contains $x$, we have $\vol^{\textout]}(L') \geq
k -\gap > \nu$, and thus we correctly output $\bot$.   
From now, we assume that 
\begin{align}  \label{eq:nu-geq-k-gap-vertex}
\nu \geq k -\gap.
\end{align}

Let the following tuple $(x',\nu', k', \gap')$ be a list of parameter
supplied to \Cref{thm:localVC_gap}. We run the local algorithm in 
\Cref{thm:localVC_gap}  using  
$(x',\nu', k', \gap') = (x, \nu, k-\gap, \gap)$. Note that
\Cref{eq:precon_vc_testing} and \Cref{eq:nu-geq-k-gap-vertex} ensure that the conditions in
\Cref{eq:precon_vc} for the local algorithm with
parameter $(x',\nu', k', \gap')$ are satisfied. The query complexity
follows from the algorithm from \Cref{thm:localVC_gap}.

\end{proof}

We present an algorithm for testing $k$-vertex connectivity and analysis. 

\paragraph{Algorithm.}  
\begin{enumerate}
\item Sample $\Theta(1)$ vertices uniformly at random. 
\item If any of the sampled vertices $x$ has out-degree less than $k$, return
  $N(x)$. 
\item Sample $\Theta(k \log k/( \epsilon \davg))$ vertices uniformly  at random
  (if $\davg$ is unknown, sample $\Theta(\log k/\epsilon)$ vertices instead). 
\item For each sampled vertex $x$, and for $i \in \{0 ,\ldots, \logk \}$, 
 \begin{enumerate}
 \item Let $\nu = 2^{i+3}\logk/\epsilon,$ and $\gap = \min\{2^{i}-1, \lfloor k/2 \rfloor \} $. 
 \item Run the local algorithm in \Cref{lem:gaplocalvc} with
   parameters $(x, \nu , k, \gap )$ on both $G$ and $G^R$ where $G^R$ is the same graph with reversed edges. %
 \end{enumerate} 
\item Return a vertex cut of size less than $k $ if any execution of
  the local algorithm above returns a vertex cut. Otherwise, declare that $G$ is $k$-vertex connected.
\end{enumerate}

\paragraph{Query Complexity.}  We first show that the number
of edge queries is at most $\ot(k^2/(\epsilon^2 \davg))$. For each
sampled vertex $x$ and $i \in \{0, \ldots, \logk\}$,
the local algorithm in 
\Cref{lem:gaplocalvc} queries $\ot( \nu k/\gap) = \ot(k/\epsilon)$
edges. The result follows as we repeat $\logk$ times per sample, and
we sample $O(k \log k /(\epsilon \davg))$ times.

\paragraph{Correctness.} If $G$ is $k$-vertex connected, it is clear that
the local algorithm \Cref{lem:gaplocalvc}  never returns any vertex cut. We show that if $G$ is
$\epsilon$-far from $k$-vertex connected, then the algorithm outputs a
vertex cut of size less than $k$ with constant probability. We start with
a simple observation. 

\begin{lemma}  \label{lem:mgeqnk-vertex}
If $m < nk/4$, then with constant probability, the algorithm outputs
a vertex cut of size less than $k$ at step 2. 
\end{lemma}
\begin{proof}
Suppose $m < nk/4$. There are at most $n/2$ vertices with out-degree at least
$k$. Hence, there are at least $n/2$ vertices of degree
less than $k$. In this case, we can sample $O(1)$ times where for each sampled vertex
$x$ we check $|N(x)| < k$ to find a $k$-vertex cut with constant probability. 
\end{proof}
From now  we assume that 
\begin{align} \label{eq:mgeqnk4-vertex2}
m \geq nk/4.
\end{align}

We start with important properties when $G$ is $\epsilon$-far from $k$-vertex connected. We say that two separation triples $(L,S,R)$ and $(L',S',R')$ are \textit{independent} if $L \cap L' = \emptyset$ or $R \cap R' = \emptyset$.  

\begin{theorem} [\cite{OrensteinR11} Corollary 17]  \label{thm:eps-far-vertex}
If a directed graph $G = (V,E)$ is $\epsilon$-far from being
$k$-vertex connected, then there exists a set
$\mathcal{F'}$ of
pairwise independent separation triples\footnote{We use the term
  separation triple $(L,S,R)$ instead of the term \textit{one-way
    pair} $(L,R)$ used by \cite{OrensteinR11} for notational
  consistency in our paper.  These terms are
  equivalent in that there is no edge from $L$ to $R$ and our $S$ is their $V - (L \cup R)$.} such that $\sum_{(L,S,R) \in \mathcal{F'}} \max\{k-|S|,0 \} > \epsilon m$.  
\end{theorem}

Let $\mathcal{F}$ be a family of pairwise independent separation
triples of $G$ such that  \\ $ p(\mathcal{F}) := \sum_{(L,S,R) \in \mathcal{F}}( \max\{k-|S|,0 \} )$ is maximized. It directly follows from \Cref{thm:eps-far-vertex} that $\sum_{(L,S,R) \in \mathcal{F}} \max\{k-|S|,0 \}  > \epsilon m$. 

We say that a left-partition $L$ of a separation triple $(L,S,R)$ is \textit{small} if $|L| \leq |R|$. Similarly, a right-partition $R$ is small if $|R| \leq |L|$. 
\begin{lemma}[\cite{FrankJ99} Lemma 7] \label{lem:LR-disjoint} 
The small left-partitions\footnote{In \cite{FrankJ99}, they use the
  term \textit{one-way pair} $(T,H)$, and define a tail $T$ of a pair $(T,H)$
if \textit{small} if $|T| \leq |H|$. Similarly, they define a head $H$ of a
pair $(T,H)$ to be \textit{small} if $|H| \leq |T|$. We only repharse
from ``tail'' to left-partition, and ``head'' to right-partition.} in $\mathcal{F}$ are pairwise disjoint, and the small right-partitions in $\mathcal{F}$ are pairwise disjoint.
\end{lemma} 

 Let $\mathcal{F}_L$ be the set of separation triples with small left-partitions in $\mathcal{F}$, and $\mathcal{F}_R$ be the set of separation triples with small-right partitions in $\mathcal{F}$.  
 By \Cref{thm:eps-far-vertex}, we have that $ \max \{ p(\mathcal{F}_L)
 , p(\mathcal{F}_R) \} > \epsilon m/2$. We assume without loss of generality that
 \begin{align} \label{eq:fl-large}
  p(\mathcal{F}_L) > \epsilon m/2. 
 \end{align}

Let $\mathcal{C}_{-1}= \{ (L,S,R) \in \mathcal{F}_L \colon |S| \geq k \}$. For $i \in \{0,\ldots,\logk\}$, let  $\mathcal{C}_i = \{ (L,S,R) \in \mathcal{F}_L \colon k - |S| \in [2^i, 2^{i+1}) \}$.  Let $\cC_{i, \textbig} = \{ (L,S,R) \in \cC_i \colon \vol^{\textout}(L) \geq  2^{i+3} \epsilon^{-1} (\logk+1) \}$, and $\cC_{i, \textsmall} = \cC_i - \cC_{i, \textbig}$. %
The following lemma is the key for the algorithm's correctness. 
\begin{lemma} \label{lemma:many-small}
There is an $i^* \in \{0, 1,\ldots, \logk \}$ such that $ |\cC_{i^*,\textsmall}| >  \epsilon n \davg
/ (8k  (\logk+1)) $. If $\davg$ is unknown, then there is an $i^* \in \{0, 1,\ldots, \logk \}$ such
that $|\cC_{i^*,\textsmall}| \geq  \epsilon n /  (32(\logk+1)) $.%
\end{lemma}

Let us briefly argue that \Cref{lemma:many-small} implies the correctness of the
algorithm. If \Cref{lemma:many-small} is true, then  by sampling
$\Theta((k \log k)/(\epsilon \davg))$ many vertices (or $\Theta(\log
k/(\epsilon))$ vertices if $\davg$ is unknown), the event
that a sampled vertex belongs to some vertex set~$L'$ in a
separation triple $(L',S',R') \in \cC_{i^*,\textsmall}$ has constant
probability (since  $\cC_{i^*,\textsmall}$ contains
pairwise disjoint small left-partitions by
\Cref{lem:LR-disjoint}).\footnote{We assumed left-partition above
  without loss of generality in the proof above, but in the algorithm
  both cases can happen.}  Now, assuming that a
sampled vertex $x$ belongs to a vertex set $L'$ as specified above. Since $L' \in
\cC_{i^*,\textsmall}$, by definition, $\vol^{\textout}(L') \leq
2^{i^*+3}\epsilon^{-1}(\lfloor \log_2 k\rfloor + 1)$, and $k -
|S'| \in [2^{i^*}, 2^{i^*+1})$. In other words, $L' \ni x, \vol^{\out}(L') \leq
2^{i^*+3}\epsilon^{-1}(\lfloor \log_2 k\rfloor + 1),$ and
$ |S'| \leq k - 2^{i^*} < k - (2^{i^*}-1)$. If $i^*$ was known, then we run the local algorithm in
\Cref{lem:gaplocalvc} with parameters $(x,\nu,k,\gap)$ where $\nu =
2^{i^*+3}\epsilon^{-1}(\lfloor \log_2 k\rfloor + 1), $ and $\gap =
\min\{2^{i^*}-1, \lfloor k/2 \rfloor \}$ to obtain a vertex cut of
size less than~$k$ with probability at least $3/4$. Since $i^* \in [0, 1, \ldots, \logk\}$, it is
enough to run the local algorithm in \Cref{lem:gaplocalvc} with
parameters $(x,\nu_i,k,\gap_i)$ where $\nu_i =
2^{i+3}\epsilon^{-1}(\logk+1)$ and $\gap_i = \min\{2^{i}-1, \lfloor
k/2 \rfloor \} $ for every $i \in \{0, 1,\ldots, \logk\}$.

It remains to check the preconditions for the local algorithm in
\Cref{lem:gaplocalvc}. We show that
\Cref{eq:precon_vc_testing} is satisfied for any call of the local
algorithm. It is easy to see that the first condition $k \geq 1+ \gap $ is
satisfied, and $k - \gap \leq n/4$ follows from \Cref{eq:input-k-testing-vertex}.  The second condition $\gap \leq k/2$ follows since
$\gap_i = \min\{2^{i}-1, \lfloor k/2 \rfloor \}  \leq \lfloor k/2
\rfloor \leq k/2$ for any $i \in \{0, 1,\ldots, \logk\}$. Finally, we show
the last condition $\nu < m(\gap+1)/(12480(k-\gap))$ is
satisfied. First, note that the for any $i \in \{0, 1,\ldots, \logk\},
\nu_i = 2^{i+3}\epsilon^{-1}(\logk+1) \leq 8\epsilon^{-1}k(\logk+1)$.
Therefore, 
\begin{multline}
\nu \leq  8\epsilon^{-1}k(\logk+1) \stackrel{~(\ref{eq:mgeqnk4-vertex2})} \leq
  8 \frac{m\epsilon^{-1}}{n} \lfloor \log_2(m/n)+1   \rfloor \stackrel{~(\ref{eq:input-k-testing-vertex})} < \frac{m}{12480k} \\ < m(\gap+1)/(12480(k-\gap))
\end{multline}
This last inequality follows since $0 \leq \gap \leq k/2$.   Therefore, the correctness follows.

\begin{proof}[Proof of \Cref{lemma:many-small}] 
First, we show that  there is an $i \in \{0, 1,\ldots, \logk \}$ such that  \begin{align} \label{eq:random2} |\cC_i| > \epsilon m /
(2^{i+2} (\logk+1)). \end{align}  To this end, we show
that there is an $i \in \{0, 1,\ldots, \logk \}$ such that 
\begin{align} \label{eq:i-satisfy-large}
\sum_{(L,S,R) \in \mathcal{C}_i} (k - |S|) > \epsilon m/ (2(\logk+1)).
\end{align}
 Suppose otherwise that for every $i \in \{0, 1,\ldots, \logk \}$, we have $\sum_{(L,S,R) \in
   \mathcal{C}_i} (k - |S|) \leq \epsilon m/ (2(\logk+1))$. Then we get
\begin{equation*}
\sum_{(L,S,R) \in \mathcal{F}_L}( \max\{k-|S|,0 \} )
 =\sum_{i=-1}^{\logk} \sum_{(L,S,R)\in \mathcal{C}_i} (k - |S|)
 = \sum_{i=0}^{\logk} \sum_{(L,S,R)\in \mathcal{C}_i} (k - |S|)
 \leq \epsilon m/2.
\end{equation*}
However, this contradicts 
 \Cref{eq:fl-large}. Second, observe that for every
$i \in \{0, 1,\ldots, \logk \}$, \begin{align} \label{eq:random-name1} |\mathcal{C}_i|2^{i+1} \geq \sum_{(L,S,R) \in \mathcal{C}_i}
 (k - |S|). \end{align} This follows trivially from that each $(L,S,R)$ in the
 set $\mathcal{C}_i$,  $k - |S| \leq 2^{i+1}$.  Therefore, for the $i \in \{0, 1,\ldots, \logk \}$ that satisfies
 \Cref{eq:i-satisfy-large}, we have 
\begin{align} \label{eq:ci-large}
|\mathcal{C}_i|  \stackrel{(\ref{eq:random-name1})} \geq
 \sum_{(L,S,R) \in \mathcal{C}_i} (k - |S|) /2^{i+1}  \stackrel{(\ref{eq:i-satisfy-large})} > \epsilon m/
 (2^{i+2} (\logk+1)). 
\end{align}

Recall that $\cC_{i, \textbig} = \{ (L,S,R) \in \cC_i \colon \vol^{\textout}(L)
\geq 2^{i+3} (\logk+1) /\epsilon \}$, and $\cC_{i, \textsmall} = \cC_i -
\cC_{i, \textbig}$. We claim that for the $i \in \{0, 1,\ldots, \logk \}$ that satisfies \Cref{eq:ci-large}, $ |\cC_{i,\textbig}| < |\cC_i|/2$. Indeed,  we have  $$  2 |C_{i,\textbig}| \leq \sum_{(L,S,R) \in \cC_{i, \textbig}} \vol^{\textout}(L)/( 2^{i+2} \epsilon^{-1} (\logk+1)) \leq m/( 2^{i+2} \epsilon^{-1} (\logk+1))  \stackrel{(\ref{eq:random2})} < |\cC_i|. $$
The first inequality follows because  $\vol^{\textout}(L)/( 2^{i+2} \epsilon^{-1} (\logk+1)) \geq 2$ for each $(L,S,R) \in \cC_{i,\textbig}$ by the definition of $ \cC_{i,\textbig} $. The second inequality follows since left-partitions in~$\cC_{i,\textbig}$ are disjoint, and $\sum_{(L,S,R) \in \cC_{i, \textbig}} \vol^{\textout}(L) \leq m$. %

Next, we have \begin{align} \label{eq:cctextsmall-vertex}
|\cC_{i,\textsmall}|
\geq |\cC_i|/2   \stackrel{(\ref{eq:random2})} > \epsilon m / (2^{i+3} (\logk+1)) \geq   \epsilon n
\davg / (8k  (\logk+1)). \end{align} The first inequality follows because $|\cC_{i,\textbig}| < |\cC_i|/2$, and $ |\cC_i| =|\cC_{i,\textbig}| +|\cC_{i,\textsmall}| $. The last inequality follows because $m = n\davg$, and $2^i \leq k$.  If $\davg$ is unknown, the last
inequality of  \Cref{eq:cctextsmall-vertex} becomes $\epsilon m / (2^{i+3} (\logk+1))  \stackrel{(\ref{eq:mgeqnk4-vertex2})}\geq   \epsilon nk
/ (32 k  (\logk+1))=\epsilon n / (32 (\logk+1)) $. 
\end{proof}

\subsection{Testing $k$-Vertex Connectivity: Bounded-Degree Model}

In this section, we prove \Cref{thm:main-property-testing} for testing $k$-vertex connectivity in the bounded degree model. By the same argument as in \Cref{sec:test-k-edge-bounded}, we assume that $G$ is $d$-regular, and thus we can sample edge uniformly by \Cref{pro:sample-edge}.

We present an algorithm for testing $k$-edge connectivity for
bounded-degree model assuming \Cref{lem:gaplocalvc}.
 
\paragraph{Algorithm.}  %
\begin{enumerate}
\item Sample $\Theta(1)$ vertices uniformly at random.
\item If any of the sampled vertices $x$ has out-degree less than $k$, return
  $N(x)$. 
\item For each $i \in \{0,\ldots,\logk\}$ and each $j \in
  \{0,\ldots, \logeta\}$ where $\eta_i = 2^{i+2}\epsilon^{-1}\logk$,
\begin{enumerate}
\item Sample $\Theta (\frac{ \logk\logeta }{ \epsilon 2^{j-i} }) = \tilde \Theta( \frac{1}{\epsilon 2^{j-i}}) $  edges uniformly at random.   
\item Let $\nu = 2^{j+1},$ and $\gap = \min\{2^{i}-1, \lfloor k/2 \rfloor \}$ . 
 \item Run the local algorithm  in \Cref{lem:gaplocalvc} with
   parameters $(x, \nu , k, \gap )$ on both $G$ and $G^R$ where
   $G^R$ is $G$ with reversed edges, and $x$ is a vertex from the
   sampled edge of the form $(x,y)$.
\end{enumerate}
\item Return a vertex cut of size less than $ k $ if any execution of
  the local algorithm above returns a vertex cut. Otherwise, declare that $G$ is $k$-vertex connected.
\end{enumerate}

\paragraph{Query Complexity.}  We first show that the number of
edge queries is at most $\ot(k/\epsilon)$. For each
vertex $x$ from the sampled edge $(x,y)$ and
for each $(i,j)$ pair in loops, the local algorithm in
\Cref{lem:gaplocalvc} 
queries $\ot( \nu k/\gap) = \ot(2^{j-i}k)$ edges, and we sample $\ot(
1/(\epsilon 2^{j-i}))$ times.  Therefore, by repeating $\ot(1)$ itereations, the total edge queries is at most $\ot( k/ \epsilon)$.

\paragraph{Correctness.} If $G$ is $k$-vertex connected, then the algorithm
never returns any vertex cut, and we are done. Suppose $G$ is
$\epsilon$-far from being $k$-vertex connected, then we show that the
algorithm outputs a vertex cut of size less than $k$ with constant
probability.  Since $G$ is $d$-regular, we have $\davg =
d$. Therefore,  we can use results from \Cref{sec:test-k-vertex}. Let $\mathcal{F}_L$ be the set of separation triples with small left-partitions in $\mathcal{F}$, and $\mathcal{F}_R$ be the set of separation triples with small-right partitions in $\mathcal{F}$.  
 By \Cref{thm:eps-far-vertex}, we have that $ \max \{ p(\mathcal{F}_L)
 , p(\mathcal{F}_R) \} > \epsilon m/2$. We assume without loss of generality that
 \begin{align} \label{eq:fl-large2}
  p(\mathcal{F}_L) > \epsilon m/2. 
 \end{align} 
Let $\mathcal{C}_{-1}= \{ (L,S,R) \in \mathcal{F}_L \colon k \leq |S| \}$. For $i \in \{0,\ldots,\logk\}$, let  $\mathcal{C}_i = \{ (L,S,R) \in \mathcal{F}_L \colon k - |S| \in [2^i, 2^{i+1}) \}$.  Let $\cC_{i, \textbig} = \{ (L,S,R) \in \cC_i \colon \vol^{\textout}(L) \geq  2^{i+3} \epsilon^{-1} \logk \}$, and $\cC_{i, \textsmall} = \cC_i - \cC_{i, \textbig}$. %
By \Cref{lemma:many-small}, there is $i$ such that  
\begin{align} \label{eq:csmallgepsilonnbard-vertex}
 |\cC_{i,\textsmall}| >  \epsilon n \davg / (8k  (\logk+1)) = \epsilon m / (8k  (\logk+1)).
\end{align}
The last inequality follows since $n\davg = nd = m$. We fix $i$ as in \Cref{eq:csmallgepsilonnbard-vertex}.  Let $\eta_i = 2^{i+3}\epsilon^{-1} \logk$. For $j \in \{0,\ldots, \logeta\}$, let $\cC_{i,\textsmall,j} = \{ (L,S,R) \in \cC_{i,\textsmall} \colon \vol^{\textout}(L) \in [2^j, 2^{j+1}) \}$. 

\begin{lemma} \label{lem:j-large-vol-vertex}
For the $i \in \{0,\ldots,\logk\}$ that satisfies \Cref{eq:csmallgepsilonnbard-vertex}, there is a $j \in \{0,\ldots, \logeta\}$ such that $\sum_{(L,S,R) \in \cC_{i,\textsmall,j}} \vol^{\textout}(L) \geq   \epsilon m 2^{j-i}/ (8 (\logk+1)( \logeta+1))$. 
\end{lemma}

Let us briefly argue that \Cref{lem:j-large-vol-vertex} implies the correctness of the algorithm.  By
sampling $\Theta (\frac{ \logk\logeta }{ \epsilon 2^{j-i} }) = \tilde
\Theta( \frac{1}{\epsilon 2^{j-i}}) $ many edges, we get
that the event that a sampled edge $(u,v)$ has $u \in L$ for some $L$ from a separation triple $(L,S,R) \in \cC_{i,\textsmall, j}$ hat constant probability (since  $\cC_{i,\textsmall}$ contains
pairwise disjoint small left-partitions by \Cref{lem:LR-disjoint}).  For each $(i,j)  \in
\{0, 1, \ldots, \logk\} \times \{0,\ldots, \logeta\}$, we run the local
algorithm in \Cref{lem:gaplocalvc}  with $\nu = 2^{j+1}$, and
$\gap = 2^i-1$; also, there exists a pair $(i,j)$ such that $\sum_{X \in \cC_{i,\textsmall,j}}
\vol^{\textout}(X) \geq \epsilon  m 2^{j-i}/ (8 (\logk+1) (\logeta+1)) $ by
\Cref{lem:j-large-vol-vertex}. Therefore,  the local algorithm in
\Cref{lem:gaplocalvc} outputs a vertex cut of size less than $k$ with
constant probability.  Note that the preconditions of the local algorithm in
\Cref{lem:gaplocalvc} can be shown to hold in a similar manner as in the previous subsection. 

\begin{proof}[Proof of \Cref{lem:j-large-vol-vertex}]
We claim that there is $j \in \{0,\ldots, \logeta\}$ such that  \begin{align} \label{eq:citextsmallj-vertex} |\cC_{i,\textsmall,j}| \geq |\cC_{i,\textsmall}| / (\logeta+1) \end{align}
Suppose otherwise. Then we have that $|\cC_{i,\textsmall,j}| < |\cC_{i,\textsmall}|/ (\logeta +1)$ for all $j \in \{0, \ldots, \logeta\} $, and therefore $\sum_{j \in
  \{0,\ldots, \logeta \} } |\cC_{i,\textsmall,j}| <
|\cC_{i,\textsmall}|$, a contradiction.

Now, for the $j \in \{0,\ldots, \logeta\}$ satisfying~\eqref{eq:citextsmallj-vertex},  we have
\begin{align*} 
 \sum_{(L,S,R) \in \cC_{i,\textsmall,j}} \vol^{\textout}(L)& \geq
|\cC_{i,\textsmall,j}|2^j \\&\stackrel{(\ref{eq:citextsmallj-vertex})} \geq
|\cC_{i,\textsmall}|2^j/ (\logeta+1) \\&
\stackrel{(\ref{eq:csmallgepsilonnbard-vertex})} \geq
\epsilon m 2^j/ (8 k (\logk+1)( \logeta+1)) \\&
 \geq \epsilon m 2^{j-i}/ (8 (\logk+1)(\logeta+1)) 
\end{align*} 
The first inequality holds because the small left-partitions in $\cC_{i,\textsmall,j}$ are pairwise disjoint by \Cref{lem:LR-disjoint} and $\vol^{\textout}(L) \geq 2^j$ by definition.  The last inequality follows since $2^i \leq k$ by definition.  
\end{proof}

	\section{Maximal $k$-Edge Connected Subgraphs}\label{sec:connected subgraphs}

In the following, we consider the problem of computing the maximal $k$-edge connected subgraphs of directed and undirected graphs.
In directed graphs, we essentially follow the overall algorithmic scheme of Chechik et al.~\cite{ChechikHILP17} and obtain an improvement by plugging in our new local cut-detection procedure.
In undirected graphs, we additionally modify the algorithmic scheme to obtain further running time improvements.

A maximal $k$-edge connected subgraph $ H $ of a graph $ G $ is a subgraph that is $k$-edge connected and there is no other subgraph of $ G $ that is $k$-edge connected and contains $ H $ as a proper subgraph.
Observe that each maximal $k$-edge connected subgraph is characterized by its set of vertices, we can thus restrict ourselves to subgraphs induced by subsets of vertices.
Furthermore, two overlapping vertex sets inducing a $k$-edge connected subgraph each can always be joined to form a larger vertex set inducing a $k$-edge connected subgraph.
Thus, the decomposition of a graph into its maximal $k$-edge connected subgraphs is unique.
To recap, the goal in this problem is to find a partition of the vertices of a graph such that (a) the subgraph induced by each vertex set of the partition is $k$-edge connected and (b) there is no vertex set that induces a $k$-edge connected subgraph and strictly contains one of the vertex sets of the partition.

\subsection{Directed Graphs}\label{sec:subgraphs directed}

The baseline recursive algorithm for computing the maximal $k$-edge connected subgraphs works as follows:
First, try to find a cut with at most $ k -1 $ cut edges.
If such a cut exists, remove the cut edges from the graph and recurse on each strongly connected component of the remaining graph.
If no such cut exists, then the graph is $k$-edge connected.
The recursion depth of this algorithm is at most $ n $, and using Gabow's cut algorithm~\cite{Gabow95}, it takes time $ O (k m \log n) $ to find a cut with at most $ k - 1 $ cut edges.
Therefore this algorithm has a running time of $ O (k m n \log n) $.

The idea of Chechik et al.\ is to speed up this baseline algorithm by using a local cut-detection procedure as follows:
The algorithm ensures that the graph contains no components with less than $ k $ outgoing edges of out-volume at most $ \nu $ and no components with less than $ k $ incoming edges of in-volume at most $ \nu $ anymore before Gabow's global cut algorithm is invoked.
This can be achieved as follows:
If the number of edges in the graph is at most~$ S (\nu) $ (an upper bound on the volume of the component returned by the local cut-detection procedure), then the basic algorithm is invoked.
Otherwise, the algorithm maintains a list~$ L $ of vertices which it considers as potential seed vertices for the local cut procedure.
Initially, $ L $ consists of all vertices.
For every vertex $ x $ of $ L $ the algorithm first tries to detect a small component containing~$ x $ and then a small containing~$ x $.
It then removes $ x $ from~$ L $ and if a component~$ C $ was detected, it removes $ C $ from the graph (as well as the outgoing and incoming edges of~$ C $) and adds the heads of the outgoing edges and the tails of the incoming edges to~$ L $.
Each component found in this way is processed (recursively) with the baseline algorithm.
Once $ L $ is empty, Gabow's cut algorithm is invoked on the remaining graph, the cut edges are removed from the graph, and the strongly connected components of the remaining graph are computed.
The algorithm then recurses on each strongly connected component with a new list $ L' $ consisting of all endpoints of the removed cut edges contained in this strongly connected component.
As a preprocessing step to this overall algorithm, we first compute the strongly connected components and run the algorithm on each strongly connected component separately.\footnote{We have added this preprocessing step to ensure that $ n = O (m) $ for each strongly connected component in the running time analysis.}

The running time analysis is as follows.
The strongly connected components computation in the preprocessing can be done in time $ O (m + n) $.
For every vertex, we initiate the local cut detection initially and whenever it was the endpoint of a removed edge.
We thus initiate at most $ O (n + m) = O (m) $ local cut detections, each taking time $ T (\nu) $.
It remains to bound the time spent for the calls of Gabow's cut algorithm and the subsequent computations of strongly connected components after removing the cut edges.
On a strongly connected graph with initially $ m' $ edges, these two steps take time $ O (k m' \log n) $.
Consider all recursive calls at the same recursion level of the algorithm.
As the graphs that these recursive calls operate on are disjoint, the total time spent at this recursion level is $ O (k m \log n) $.
To bound the total recursion depth, observe that for a graph with initially $ m' $ edges, the graph passed to each recursive call has at most $ \max \{ m' - \nu, S (\nu) \} $ edges as the only cuts left to find for Gabow's cut algorithm either have in- or out-volume at least $ \nu $ on one side of the cut.
Thus, the recursion depth is $ O (\tfrac{m}{\nu} + S (\nu)) $.
Altogether, we therefore arrive at a running time of
\begin{equation*}
O \left(m \cdot T (\nu) + \left( \frac{m}{\nu} + S(\nu) \right) \cdot k m \log n + n \right) \, .
\end{equation*}

Observe further that a one-sided Monte-Carlo version of the local cut-detection procedure, as the one we are giving in this paper, only affects the recursion depth.
If each execution of the procedure is successful with probability $ p \geq 1 - \tfrac{1}{n^3} $, then the probability that all $ O (m) = O (n^2) $ executions of the procedure are successful is at least $ 1 - O(\tfrac{1}{n}) $.
As the worst-case recursion depth is at most~$ n $, the expected recursion depth is at most $ O ((1 - \tfrac{1}{n}) \cdot (\tfrac{m}{\nu} + S (\nu)) + \tfrac{1}{n} \cdot n) = O (\tfrac{m}{\nu} + S (\nu)) $.

The analysis of this algorithmic scheme can be summarized in the following lemma.
\begin{lemma}[Implicit in \cite{ChechikHILP17}]\label{lem:reduction maximal connected subgraphs}
Suppose there is an algorithm that, given constant-time query access to a directed graph, for a fixed $ k \geq 2 $, any integer~$ \nu \geq 1 $ and any seed vertex~$ x $, runs in time $ T (\nu) $ and has the following behavior:
(1) If there is a set $ L \subseteq V $ containing~$ x $ with $ | E (L, V - L) | < k $ of volume at most~$ \nu $, then the algorithm returns such a set of volume at most $ S (\nu) $ with probability at least $ 1 - \tfrac{1}{n^3} $ (and $ \bot $ otherwise) and
(2) if no such set exists, then the algorithm returns $ \bot $.

Then there is an algorithm for computing the maximal $k$-edge connected subgraphs of a directed graph with expected running time $ O (m \cdot T (\nu) + (\tfrac{m}{\nu} + S(\nu)) \cdot k m \log n + n) $ for every $ 1 \leq \nu \leq m $.
\end{lemma}

For any $ \nu \geq 1 $, our improved local cut-detection procedure of Corollary~\ref{cor:localEC_exact_new} has $ T (\nu) = O (k^2 \nu) $ and $ S (\nu) = O (k \nu) $ with success probability $ \tfrac{1}{2} $.\footnote{Note that with this choice of $ \nu $ the precondition $ \nu < m / (130 k) $ of Corollary~\ref{cor:localEC_exact_new} holds for any $ k < m / 130^2 $. For larger $ k $, we can simply return $ V $ to satisfy the conditions of~(1) in Lemma~\ref{lem:reduction maximal connected subgraphs}.}
We can boost the success probability to $ 1 - \tfrac{1}{n^3} $ by repeating the algorithm $ \lceil \log (3 n) \rceil $ times, which leads to $ T (\nu) = O (k^2 \nu \log n) $.
By setting $ \nu = \tfrac{\sqrt{m}}{\sqrt{k}} $, we arrive at running time of
\begin{equation*}
O \left( k^2 m \nu \log n + \left( \frac{m}{\nu} + k \nu \right) \cdot k m \log n + n \right) = O (k^{3/2} m^{3/2} \log n + n) \, .
\end{equation*}

\begin{theorem}
There is a randomized Las Vegas algorithm for computing the maximal $k$-edge connected subgraphs of a directed graph that has expected running time $ O (k^{3/2} m^{3/2} \log n + n) $.
\end{theorem}

\subsection{Undirected Graphs}

In undirected graphs, we obtain a tighter upper bound on the running time of the algorithm for three reasons.
First, instead of parameterizing the algorithm by a target volume $ \nu $, we parameterize it by a target vertex size~$ \sigma $.
Thus, whenever the list $ L $ in the algorithm becomes empty we can be sure that there is no component of vertex size at most $ \sigma $ in the current graph anymore.
Components that are larger can be found at most $ \tfrac{n}{\sigma} $ times.
Second, for each component detected by the local cut procedure the number of incoming edges equals its number of outgoing edges and is at most $ k - 1 $.
Thus, the number of removed edges per successful component detection is at most $ k - 1 $.
As there are at most $ n $ such successful detections in total, the total number of executions of the local cut-detection procedure is at most $ n + (k-1) n = k n $.
Third, in undirected graphs we can run instances of both Gabow's global cut detection algorithm and our local cut detection algorithm on a sparse $k$-edge connectivity certificate~\cite{Thurimella89,NagamochiI92} of the current graph.
The sparse certificate can be computed in time $ O (m + n) $, but we do not want to explicitly perform this expensive computation each time edges are removed from the graph. %
Instead, we maintain the sparse certificate with a dynamic algorithm as outlined below.\footnote{Let us emphasize again that the sparse certificate is not w.r.t.\ to the input graph but w.r.t.\ to the graph in which all cut edges found by Gabow's algorithm and all edges incident on components detected by the local algorithm have been removed.}
As the $k$-edge connectivity certificate has $ O (k n) $ edges, Gabow's cut algorithm has running time $ O (k^2 n \log{n}) $ if it is run on the certificate.
Furthermore, as observed by Nanongkai et al.~\cite{NanongkaiSY19}, the $k$-edge connectivity certificate has arboricity $ k $.
Since the arboricity bounds the local density of any vertex-induced subgraph, the volume of a component of vertex size $ \sigma $ is $ O (\sigma k) $ on the $k$-edge connectivity certificate.
Each instance of the local cut detection procedure therefore has running time $ T (k \sigma) $ and if successful detects a component of volume $ S (k \sigma) $.

A sparse $k$-edge connectivity certificate of a graph $ G = (V, E) $ is a graph $ H = (V, F_1 \cup \ldots \cup F_k) $ such that for every $ 1 \leq i \leq k $ the graph $ (V, F_i) $ is a spanning forest of $ (V, E - \bigcup_{1 \leq j \leq i- 1} F_j) $ (where in particular $ (V, F_1) $ is a spanning forest of $ G $).
The dynamic connectivity algorithm of Holm et al.~\cite{HolmLT01} can be used to dynamically maintain a spanning forest of a graph undergoing edge insertions and deletions in time $ O (\log^2 n) $ per update.
Maintaining the hierarchy of spanning forests for a $k$-edge connectivity certificate under a sequence edge removals to $ G $ takes time $ O (k m \log^2 {n}) $ by the following argument\footnote{This argument is similar to the one for maintaining the cut sparsifier in~\cite{AbrahamDKKP16}.}:
The dynamic algorithm of Holm et al.\ makes at most one change to the spanning forest per change to the input graph.
Therefore each deletion in $ G $ causes at most one update to to each of the $ k $ levels of the hierarchy.
As at most $ m $ edges can be removed from~$ G $, the total update time is $ O (k m \log^2 {n}) $.

Using otherwise the same analysis as in Section~\ref{sec:subgraphs directed}, we arrive at the following running time:
\begin{equation*}
O \left(k n \cdot T (k \sigma) + \left( \frac{n}{\sigma} + S(k \sigma) \right) \cdot k^2 n \log n + k m \log^2 n \right) \, .
\end{equation*}
The analysis of this algorithmic scheme can be summarized in the following lemma.
\begin{lemma}
Suppose there is an algorithm that, given constant-time query access to an undirected graph, for a fixed $ k \geq 2 $, any integer~$ \nu \geq 1 $ and any seed vertex~$ x $, runs in time $ T (\nu) $ and has the following behavior:
(1) If there is a set $ L \subseteq V $ containing~$ x $ with $ | E (L, V - L) | < k $ of volume at most~$ \nu $, then the algorithm returns such a set of volume at most $ S (\nu) $ with probability at least $ 1 - \tfrac{1}{n^3} $ (and $ \bot $ otherwise) and
(2) if no such set exists, then the algorithm returns $ \bot $.

Then there is a randomized Las Vegas algorithm for computing the maximal $k$-edge connected subgraphs of an undirected graph with expected running time $ O (k n \cdot T (k \sigma) + (\tfrac{n}{\sigma} + S(k \sigma)) \cdot k^2 n \log n + k m \log^2 n) $ for every $ 1 \leq \sigma \leq n $.
\end{lemma}

For any $ \sigma \geq 1 $, our local cut-detection procedure of Corollary~\ref{cor:localEC_exact_new} has $ T (k \sigma) = O (k^3 \sigma) $ and $ S (k \sigma) = O (k^2 \sigma) $ with success probability $ \tfrac{1}{2} $.
We can boost the success probability to $ 1 - \tfrac{1}{n^3} $ by repeating the algorithm $ \lceil \log (3 n) \rceil $ times, which leads to $ T (\nu) = O (k^3 \sigma \log n) $.
Thus, the running time is
\begin{equation*}
O \left(k^4 n \sigma \log n + \left( \frac{n}{\sigma} + k^2 \sigma \right) \cdot k^2 n \log n + k m \log^2 n \right) \, .
\end{equation*}
By setting $ \sigma = \tfrac{\sqrt{n}}{k} $, we obtain a running time of $ O (k^3 n^{3/2} \log{n} + k m \log^2 n) $.

\begin{theorem}
There is a randomized Las Vegas algorithm for computing the maximal $k$-edge connected subgraphs of an undirected graph that has expected running time $ O (k^3 n^{3/2} \log{n} + k m \log^2 n) $.
\end{theorem}

\paragraph{Remark on Sparse Certificates}

Our algorithm tailored to undirected graphs relies on sparse certificates.
It would be tempting to run our algorithm for directed graphs directly on a sparse certificate to achieve running-time savings by performing a maximum amount of sparsification.
However this approach does not work as Figure~\ref{fig:sparse certificate} shows.
Sparse certificates preserve the $k$-edge connectivity, but not necessarily the maximal $k$-edge connected subgraphs.

\begin{figure}[htbp!]
    \centering
    \begin{tikzpicture}
	\tikzset{default_node/.style={circle,draw,fill=black,inner sep=0pt,minimum size=8pt}}
	\tikzset{edge/.style={thick,black}}
	
	\draw[rounded corners,fill=yellow!50] (-0.5,-2) rectangle (2,0.5);
	
	\node[default_node] (1) at (0,-1.5) {};
	\node[default_node] (2) at (0,0) {};
	\node[default_node] (3) at (1.5,0) {};
	\node[default_node] (4) at (1.5,-1.5) {};
	
	\node[default_node] (5) at (-1.5,-0.75) {};
	\node[default_node] (6) at (0.75,1.5) {};
	\node[default_node] (7) at (3.0,-0.75) {};

	\draw[edge] (1) edge (2);
	\draw[edge] (1) edge (3);
	\draw[edge] (1) edge (4);
	\draw[edge] (2) edge (3);
	\draw[edge] (2) edge (4);
	\draw[edge] (3) edge (4);
	
	\draw[edge] (5) edge (1);
	\draw[edge] (5) edge (2);
	\draw[edge] (6) edge (2);
	\draw[edge] (6) edge (3);
	\draw[edge] (7) edge (3);
	\draw[edge] (7) edge (4);
    \end{tikzpicture}
    \hspace{1em}
    \begin{tikzpicture}
	\tikzset{default_node/.style={circle,draw,fill=black,inner sep=0pt,minimum size=8pt}}
	\tikzset{edge/.style={thick,black}}
	\tikzset{spanning_tree_edge1/.style={ultra thick,red!75}}
	\tikzset{spanning_tree_edge2/.style={ultra thick,green!75!black}}
	\tikzset{spanning_tree_edge3/.style={ultra thick,blue!75}}

	\draw[rounded corners,draw=none] (-0.5,-2) rectangle (2,0.5);

	\node[default_node] (1) at (0,-1.5) {};
	\node[default_node] (2) at (0,0) {};
	\node[default_node] (3) at (1.5,0) {};
	\node[default_node] (4) at (1.5,-1.5) {};
	
	\node[default_node] (5) at (-1.5,-0.75) {};
	\node[default_node] (6) at (0.75,1.5) {};
	\node[default_node] (7) at (3.0,-0.75) {};

	\draw[spanning_tree_edge2] (1) edge (2);
	\draw[spanning_tree_edge2] (1) edge (3);
	\draw[spanning_tree_edge2] (1) edge (4);
	\draw[spanning_tree_edge3] (2) edge (3);
	\draw[spanning_tree_edge3] (2) edge (4);
	
	\draw[spanning_tree_edge1] (5) edge (1);
	\draw[spanning_tree_edge1] (5) edge (2);
	\draw[spanning_tree_edge1] (6) edge (2);
	\draw[spanning_tree_edge1] (6) edge (3);
	\draw[spanning_tree_edge1] (7) edge (3);
	\draw[spanning_tree_edge1] (7) edge (4);
    \end{tikzpicture}
    
    \vspace{2ex}
    
    \begin{tikzpicture}[scale=0.8]
	\tikzset{default_node/.style={circle,draw,fill=black,inner sep=0pt,minimum size=8pt}}
	\tikzset{edge/.style={thick,black}}
	\tikzset{spanning_tree_edge1/.style={ultra thick,red!75}}
	\node[default_node] (1) at (0,-1.5) {};
	\node[default_node] (2) at (0,0) {};
	\node[default_node] (3) at (1.5,0) {};
	\node[default_node] (4) at (1.5,-1.5) {};
	
	\node[default_node] (5) at (-1.5,-0.75) {};
	\node[default_node] (6) at (0.75,1.5) {};
	\node[default_node] (7) at (3.0,-0.75) {};

	\draw[edge] (1) edge (2);
	\draw[edge] (1) edge (3);
	\draw[edge] (1) edge (4);
	\draw[edge] (2) edge (3);
	\draw[edge] (2) edge (4);
	\draw[edge] (3) edge (4);
	
	\draw[spanning_tree_edge1] (5) edge (1);
	\draw[spanning_tree_edge1] (5) edge (2);
	\draw[spanning_tree_edge1] (6) edge (2);
	\draw[spanning_tree_edge1] (6) edge (3);
	\draw[spanning_tree_edge1] (7) edge (3);
	\draw[spanning_tree_edge1] (7) edge (4);
    \end{tikzpicture}
    \begin{tikzpicture}[scale=0.8]
	\tikzset{default_node/.style={circle,draw,fill=black,inner sep=0pt,minimum size=8pt}}
	\tikzset{edge/.style={thick,black}}
	\tikzset{spanning_tree_edge1/.style={thick,dashed,black!50}}
	\tikzset{spanning_tree_edge2/.style={ultra thick,green!75!black}}
	\node[default_node] (1) at (0,-1.5) {};
	\node[default_node] (2) at (0,0) {};
	\node[default_node] (3) at (1.5,0) {};
	\node[default_node] (4) at (1.5,-1.5) {};
	
	\node[default_node] (5) at (-1.5,-0.75) {};
	\node[default_node] (6) at (0.75,1.5) {};
	\node[default_node] (7) at (3.0,-0.75) {};

	\draw[spanning_tree_edge2] (1) edge (2);
	\draw[spanning_tree_edge2] (1) edge (3);
	\draw[spanning_tree_edge2] (1) edge (4);
	\draw[edge] (2) edge (3);
	\draw[edge] (2) edge (4);
	\draw[edge] (3) edge (4);
	
	\draw[spanning_tree_edge1] (5) edge (1);
	\draw[spanning_tree_edge1] (5) edge (2);
	\draw[spanning_tree_edge1] (6) edge (2);
	\draw[spanning_tree_edge1] (6) edge (3);
	\draw[spanning_tree_edge1] (7) edge (3);
	\draw[spanning_tree_edge1] (7) edge (4);
    \end{tikzpicture}
    \begin{tikzpicture}[scale=0.8]
	\tikzset{default_node/.style={circle,draw,fill=black,inner sep=0pt,minimum size=8pt}}
	\tikzset{edge/.style={thick,black}}
	\tikzset{spanning_tree_edge1/.style={thick,dashed,black!50}}
	\tikzset{spanning_tree_edge2/.style={thick,dashed,black!50}}
	\tikzset{spanning_tree_edge3/.style={ultra thick,blue!75}}

	\node[default_node] (1) at (0,-1.5) {};
	\node[default_node] (2) at (0,0) {};
	\node[default_node] (3) at (1.5,0) {};
	\node[default_node] (4) at (1.5,-1.5) {};
	
	\node[default_node] (5) at (-1.5,-0.75) {};
	\node[default_node] (6) at (0.75,1.5) {};
	\node[default_node] (7) at (3.0,-0.75) {};

	\draw[spanning_tree_edge2] (1) edge (2);
	\draw[spanning_tree_edge2] (1) edge (3);
	\draw[spanning_tree_edge2] (1) edge (4);
	\draw[spanning_tree_edge3] (2) edge (3);
	\draw[spanning_tree_edge3] (2) edge (4);
	\draw[edge] (3) edge (4);
	
	\draw[spanning_tree_edge1] (5) edge (1);
	\draw[spanning_tree_edge1] (5) edge (2);
	\draw[spanning_tree_edge1] (6) edge (2);
	\draw[spanning_tree_edge1] (6) edge (3);
	\draw[spanning_tree_edge1] (7) edge (3);
	\draw[spanning_tree_edge1] (7) edge (4);
    \end{tikzpicture}

    \caption{The top row shows an example of an undirected graph on the left with maximal $3$-edge component consisting of four vertices (marked by the yellow box) that is not preserved in its sparse certificate on the right. The bottom row shows the $ 3 $~spanning forests found for the sparse certificate.}
    \label{fig:sparse certificate}
\end{figure}
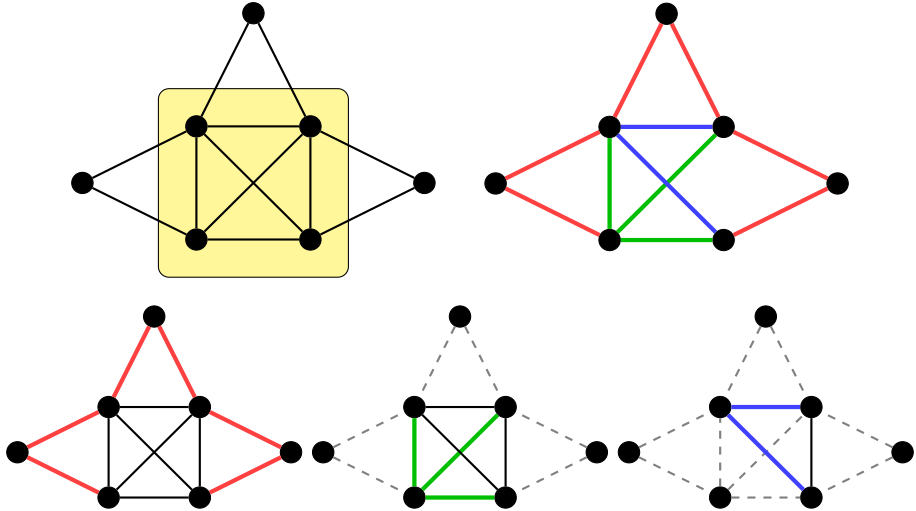

	\section*{Acknowledgement}  
        
		S.~Forster and L.~Yang would like to thank Asaf Ferber for fruitful discussions.
		S.~Forster would additionally like to thank Gramoz Goranci for giving comments on a draft related to this paper.
        This project has received funding from the European Research
        Council (ERC) under the European Union's Horizon 2020 research
        and innovation programme under grant agreement No
        715672 and 759557. Nanongkai was also partially supported by the Swedish
        Research Council (Reg. No. 2015-04659.).	
	
	 \appendix  
        
\section{An Alternative Algorithm for Local Edge Connectivity}

\label{sec:localEC_alternate}

In this section, we give another local algorithm for detecting an edge cut
of size $k$ and volume~$\nu$ containing some seed vertex in time $O(\nu k^{2})$.
Both algorithm and analysis are very simple. 
\begin{theorem}
\label{thm:localEC_approx}There is a randomized (Monte Carlo) algorithm
that takes as input a vertex $x\in V$ of an $n$-vertex $m$-edge
graph $G=(V,E)$ represented as adjacency lists, a volume parameter
$\nu$, a cut-size parameter $k\ge1$, and an accuracy parameter $\epsilon\in(0,1]$
where $\nu<\epsilon m/8$ and runs in $O(\nu k/\epsilon)$ time and
 outputs either
\begin{itemize}
\item the symbol ``$\bot$'' indicating that, with probability $1/2$,
there is no $S\ni x$ where $|E(S,V-S)|<k$ and $\vol^{\out}(S)\le\nu$,
or
\item a set $S\ni x$ where $S\neq V$, $|E(S,V-S)|<\left\lfloor (1+\epsilon)k\right\rfloor $
and $\vol^{\out}(S)\le10\nu/\epsilon$.\footnote{We note that the factor 10 in \Cref{thm:localEC_approx} can be improved.
We only use this factor for simplifying the analysis.}
\end{itemize}
\end{theorem}

By setting $\epsilon<\frac{1}{k}$, we have that $\left\lfloor (1+\epsilon)k\right\rfloor =k$.
In particular, we obtain an algorithm for the \emph{exact} problem:
\begin{corollary}
\label{thm:localEC_exact}There is a randomized (Monte Carlo) algorithm
that takes as input a vertex $x\in V$ of an $n$-vertex $m$-edge
graph $G=(V,E)$ represented as adjacency lists, a volume parameter
$\nu$, and a cut-size parameter $k\ge1$ where $\nu<m/8k$ and runs
in $O(\nu k^{2})$ time and outputs either
\begin{itemize}
\item the symbol ``$\bot$'' indicating that, with probability $1/2$,
there is no $S\ni x$ where $|E(S,V-S)|<k$ and $\vol^{\out}(S)\le\nu$,
or
\item a set $S\ni x$ where $S\neq V$, $|E(S,V-S)|<k$ and $\vol^{\out}(S)\le10\nu k$.
\end{itemize}
\end{corollary}

\begin{algorithm}
\SetKwFor{RepTimes}{repeat}{times}{end}

\RepTimes{$\left\lfloor (1+\epsilon)k\right\rfloor $}
{
	Grow a DFS tree $T$ starting from $x$ and stop once exactly $8\nu/\epsilon$ edges have been visited. 
	
	Let $E_{\DFS}$ be the set of edges visited.
	
	\lIf{$|E_{\DFS}|<8\nu/\epsilon$}{\Return $V(T)$.}
	\Else{
		Sample an edge $(y',y)\in E_{\DFS}$ uniformly.
		 
		Reverse the direction of edges in the path $P_{xy}$ in $T$ from $x$ to $y$.
	}
}
\Return{$\bot$.}

\caption{$\localEC(x,\nu,k,\epsilon)$\label{alg:local}}
\end{algorithm}

The algorithm for \Cref{thm:localEC_approx} in described in \Cref{alg:local}.
We start with the following important observation.
\begin{lemma}
\label{lem:reverse_cutsize}Let $S\subset V$ be any set where $x\in S$.
Let $P_{xy}$ be a path from $x$ to $y$. Suppose we reverse the
direction of edges in $P_{xy}$. Then, we have $|E(S,V-S)|$ and $\vol^{\out}(S)$
are both decreased exactly by one if $y\notin S$. Otherwise, $|E(S,V-S)|$
and $\vol^{\out}(S)$ stay the same.
\end{lemma}

It is clear that running time of \Cref{alg:local} is $\left\lfloor (1+\epsilon)k\right\rfloor \times O(\nu/\epsilon)=O(\nu k/\epsilon)$ 
because the DFS tree only requires $O(\nu/\epsilon)$ for visiting $O(\nu/\epsilon)$ edges. 
The two lemmas below imply the correctness of \Cref{thm:localEC_approx}
\begin{lemma}
If a set $S$ is returned, then $S\ni x$, $S\neq V$, $|E(S,V-S)|<\left\lfloor (1+\epsilon)k\right\rfloor $
and $\vol^{\out}(S)\le10\nu/\epsilon$.
\end{lemma}

\begin{proof}
If $S$ is returned, then the DFS tree $T$ get stuck at $S=V(T)$.
That is, $|E(S,V-S)|=0$ and $\vol^{\out}(S)\le8\nu/\epsilon$ at
the end of the algorithm.  Note that $x\in S$ and $S\neq V$ because
$8\nu/\epsilon<m$. Observe that the algorithm has reversed strictly
less than $\left\lfloor (1+\epsilon)k\right\rfloor $ many paths $P_{xy}$,
because the algorithm did not reverse a path in the iteration that
$S$ is returned. So \Cref{lem:reverse_cutsize} implies that, initially,
$|E(S,V-S)|<\left\lfloor (1+\epsilon)k\right\rfloor $ and, $\vol^{\out}(S)< 8\nu/\epsilon+\left\lfloor (1+\epsilon)k\right\rfloor \le10\nu/\epsilon$.
\end{proof}
\begin{lemma}
If $\bot$ is returned, then, with probability at least $1/2$, there
is no $S\ni x$ where $|E(S,V-S)|<k$ and $\vol^{\out}(S)\le\nu$.
\end{lemma}

\begin{proof}
Suppose that such $S$ exists. We will show that $\bot$ is returned
with probability less than $1/2$. Suppose that no set $S'$ is returned
before the last iteration. It suffices to show that at the beginning
of the last iteration, $|E(S,V-S)|=0$ with probability at least $1/2$.
If this is true, then the DFS tree $T$ in the last iteration will
not be able to visit more than $\nu$ edges and so will return the
set $V(T)$.

Let $k'=\left\lfloor (1+\epsilon)k\right\rfloor -1$ denote the number
of iterations excluding the last one. Let $X_{i}$ be the random variable
where $X_{i}=1$ if the sampled edge $(y',y)$ in the $i$-th iteration
of the algorithm satisfies $y\in S$. Otherwise, $X_{i}=0$. As
$\vol^{\out}(S)$ never increases, observe that $\Exp [X_{i}]\le\frac{\vol^{\out}(S)}{|E_{\DFS}|}\le\frac{\nu}{8\nu/\epsilon}=\epsilon/8$
for each $i\le k'$. Let $X=\sum_{i=1}^{k'}X_{i}$. We have $\Exp [X]\le\epsilon k'/8$
by linearity of expectation and $\Pr[X\le\epsilon k'/4]\ge1/2$ by
Markov's inequality. So $\Pr[X\le\left\lfloor \epsilon k'/4\right\rfloor ]\ge1/2$
as $X$ is integral.

Let $Y=k'-X$. Notice that $Y$ is the number of times before the
last iteration where the algorithm samples $y\notin S$. We claim
that $k'-\left\lfloor \epsilon k'/4\right\rfloor \ge k-1$ (see the
proof at the end). Hence, with probability at least $1/2$, $Y\ge k'-\left\lfloor \epsilon k'/4\right\rfloor \ge k-1\ge|E(S,V-S)|$.
By \Cref{lem:reverse_cutsize}, if $Y\ge|E(S,V-S)|$, then $|E(S,V-S)|=0$
at the beginning of the last iteration. This concludes the proof.
\begin{claim}
$k'-\left\lfloor \epsilon k'/4\right\rfloor \ge k-1$ for $\epsilon\in[0,1]$
\end{claim}

\begin{proof}
If $\epsilon<4/k'$, then $\left\lfloor \epsilon k'/4\right\rfloor =0$,
so $k'-\left\lfloor \epsilon k'/4\right\rfloor =\left\lfloor (1+\epsilon)k\right\rfloor -1\ge k-1$.
If $\epsilon\ge4/k'$, then\footnote{The main reason we choose the factor $8$ in the number $8\nu/\epsilon$
of visited edges by the DFS is for simplifying the following inequalities.}
\begin{align*}
k'-\left\lfloor \epsilon k'/4\right\rfloor  & \ge(1-\epsilon/4)k'\\
 & \ge(1-\epsilon/4)((1+\epsilon)k-2)\\
 & \ge(1-\epsilon/4)(1+\epsilon)k-2\\
 & \ge(1+\epsilon/2)k-2\\
 & \ge k-1.
\end{align*}
where the last inequality is because $\epsilon k/2\ge\frac{4}{k'}\cdot\frac{k}{2}\ge1$
as $k'\le\left\lfloor (1+\epsilon)k\right\rfloor \le2k$.
\end{proof}
\end{proof}

	\ifdefined\AdvCite

	\printbibliography[heading=bibintoc] 
	
	\else 
 	
	\bibliographystyle{alpha}
	\bibliography{references} 
	
	\fi

\end{document}